\documentclass[journal,twocolumn]{IEEEtran}
\IEEEoverridecommandlockouts                         

\usepackage[noadjust]{cite}
\usepackage[utf8]{inputenc} 
\usepackage[T1]{fontenc} 
\usepackage{hyperref}       
\usepackage{url}            
\usepackage{booktabs}      
\usepackage{amsfonts}      
\usepackage{nicefrac}      
\usepackage{microtype}    
\usepackage{algorithmic}
\usepackage[ruled,vlined]{algorithm2e}
\usepackage{amsmath}
\usepackage{graphicx}
\usepackage{multicol}
\usepackage[font=footnotesize, labelfont=bf]{caption}

\usepackage{amsthm}
\newtheorem{theorem}{Theorem}

\newtheorem{lemma}[theorem]{Lemma}
\newtheorem*{remark}{Remark}

\usepackage{bbm}
\usepackage{dblfloatfix}
\usepackage{subcaption,graphicx}
\usepackage{hyperref}
\hypersetup{colorlinks=true,linkcolor=blue,urlcolor=blue}

\usepackage[]{algorithm2e}
\SetKwInput{KwInput}{Input \hspace*{0.3em}}
\SetKwInput{KwOutput}{Output}
\usepackage{tikz}

\usepackage{amssymb}
\usepackage{pgfplots}
\usepackage{pgfmath}
\usepgfplotslibrary{patchplots}
\usetikzlibrary{patterns, positioning, arrows}
\usepackage[shortlabels]{enumitem}
\pgfmathsetmacro\sprayRadius{.75pt}
\pgfmathsetmacro\sprayPeriod{.8cm}

\usepackage{chngcntr}
\counterwithin{paragraph}{section}

\usepackage[flushleft]{threeparttable}

\newcommand\copyrighttext{
  \footnotesize \textcopyright 2023 IEEE. Personal use of this material is permitted.
  Permission from IEEE must be obtained for all other uses in any current or future
  media, including reprinting/republishing this material for advertising or promotional
  purposes, creating new collective works, for resale or redistribution to servers or
  lists, or reuse of any copyrighted component of this work in other works.
  DOI: \href{<http://tex.stackexchange.com>}{10.1109/TNNLS.2023.3307470}}
\newcommand\copyrightnotice{%
\begin{tikzpicture}[remember picture,overlay]
\node[anchor=south,yshift=3pt] at (current page.south) {\fbox{\parbox{\dimexpr\textwidth-\fboxsep-\fboxrule\relax}{\copyrighttext}}};
\end{tikzpicture}%
}

\title{
Graph Neural Networks on SPD Manifolds \\
for Motor Imagery Classification:\\
A Perspective from the Time-Frequency Analysis
}

\author{
Ce~Ju and Cuntai~Guan~\IEEEmembership{Fellow, IEEE}
\thanks{Ce Ju and Cuntai Guan are with the S-Lab and School of Computer Science and Engineering, Nanyang Technological University, 50 Nanyang Avenue, Singapore (emails: \{juce0001,ctguan\}@ntu.edu.sg). 
}
}
\begin{document}

\maketitle
\thispagestyle{empty}
\pagestyle{empty}

\copyrightnotice

\begin{abstract}
The motor imagery (MI) classification has been a prominent research topic in brain-computer interfaces based on electroencephalography (EEG). Over the past few decades, the performance of MI-EEG classifiers has seen gradual enhancement. In this study, we amplify the geometric deep learning-based MI-EEG classifiers from the perspective of time-frequency analysis, introducing a new architecture called Graph-CSPNet. We refer to this category of classifiers as \emph{Geometric Classifiers}, highlighting their foundation in differential geometry stemming from EEG spatial covariance matrices. Graph-CSPNet utilizes novel manifold-valued graph convolutional techniques to capture the EEG features in the time-frequency domain, offering heightened flexibility in signal segmentation for capturing localized fluctuations. To evaluate the effectiveness of Graph-CSPNet, we employ five commonly-used publicly available MI-EEG datasets, achieving near-optimal classification accuracies in nine out of eleven scenarios. The Python repository can be found at \url{github.com/GeometricBCI/Tensor-CSPNet-and-Graph-CSPNet}.
\end{abstract}


\IEEEpeerreviewmaketitle

\section{Introduction}
A brain-computer interface (BCI) is a technology that measures and analyzes the relevant information of a user's brain activity and establishes a communication link between the brain and the external environment~\cite{wolpaw2000brain}.
Electroencephalogram (EEG)-based BCIs are among the most widely used, portable, and cost-effective BCIs, and have led to numerous applications, such as post-stroke motor rehabilitation, control of a wheelchair system, and video gaming~\cite{yuan2014brain}.
The control of these EEG-based BCI applications is accompanied by EEG rhythmic changes over the sensorimotor cortices, including the posterior frontal and anterior parietal regions, and a wealth of studies and evidence over the past thirty years has demonstrated that the sensorimotor rhythm changes associated with motor imagery (MI) are effective control signals for BCIs~\cite{pfurtscheller1999event,pfurtscheller2001motor,pfurtscheller201213}.

\begin{table*}[!t]
\caption{Comparison between Tensor-CSPNet and Graph-CSPNet.~\label{tab:tensor-graph}
}
\centering 
\begin{tabular}{l  l  l  }
\toprule
Geometric Classifiers  & Tensor-CSPNet & Graph-CSPNet (Proposed) \\
 \midrule
Network Input & Tensorized Spatial Covariance Matrices. & Time-Frequency Graph.\\
Architecture &BiMaps; CNNs. & Graph-BiMaps.\\
Distinctive Structure &CNNs for Temporal Dynamics.& Spectral Clustering for Time-Frequency Distributions.\\
Training Optimizer & Riemannian Optimization & Riemannian Optimization.\\
 \midrule
Underlying Space & ($\mathcal{S}_{++}, g^{AIRM}$).  & ($\mathcal{S}_{++}, g^{AIRM}$). \\
Methodology Heritage & CSP. & CSP; Riemannian-Based Approaches.\\
Design Principle &The Time-Space-Frequency Principle: &The Time-Space-Frequency Principle and the Principle of Time-Frequency\\
{} &\emph{Exploitation in the frequency, space, and time } & Analysis: \emph{Exploitation in the time-frequency domain simultaneously, and} \\
{} &\emph{domains sequentially.} & {then in the space domain.} \\
\bottomrule
\end{tabular}
\end{table*}

During the planning and execution of movement, the sensorimotor rhythms exhibit changes in amplitude that are referred to as event-related desynchronization (ERD) and event-related synchronization (ERS), corresponding to decreases and increases in rhythmic activity, respectively. 
These changes generate patterns that can be reliably and distinctly discerned, enabling the EEG signals to be accurately classified~\cite{pfurtscheller1997eeg}. 
In an EEG-based motor imagery (MI) task, the individual mentally simulates physical movement, which activates the cortical sensorimotor systems~\cite{jeannerod1994representing} and primary sensorimotor areas~\cite{beisteiner1995mental}, and an EEG device records EEG signals. 
Various machine learning classifiers are then utilized to decode an individual’s intentions~\cite{lotte2018review}.

Building on the ERD/ERS effect observed during a MI task, conventional MI-EEG classifiers adopt the \emph{time-space-frequency principle} to extract EEG signal patterns and rhythms. 
The essence of this principle lies in the analysis of EEG signals in terms of frequency, time, and space, allowing for the identification of patterns across these dimensions based on the dominant frequencies of rhythmic activity, their periods of occurrence, and their distribution over the sensorimotor cortices.~\cite{qin2005classification,ferree2009space} 
This principle has been generalized to capture patterns across frequency, time, and space in EEGs.
However, the non-stationary nature of EEG spectral contents can hinder the effectiveness of conventional methods that rely on statistical stationarity assumptions~\cite{allen2010time,sanei2013eeg}. 
In these circumstances, traditional Fourier analysis proves to be ineffective.
To address this challenge, the time-frequency analysis is employed to localize rhythmic components in real time, strengthening the MI-EEG classifiers~\cite{blanco1995time,roach2008event,tzallas2009epileptic,morales2022time}. 
Examples of such methods include the short-term frequency transform and the wavelet transform~\cite{adeli2003analysis,subasi2007eeg,adeli2007wavelet}, which have been integrated into the CSP approach to creating novel wavelet-based classifiers such as the wavelet-CSP classifier~\cite{robinson2013eeg}.

Motivated by the principles of Gabor's time-frequency theory and Morlet's wavelet theory~\cite{gabor1946theory,morlet1982wave_a, morlet1982wave_b}, this study aims to enhance a geometric deep learning-based MI-EEG classifier, namely Tensor-CSPNet~\cite{ju2022tensor}, to facilitate a more effective exploration of local oscillatory components in EEGs.
The proposed approach employs graph-based neural networks, Graph-CSPNet, to handle the MI-EEG classification using EEG spatial covariance matrices.
A notable characteristic of this neural network incorporates a novel graph BiMap layer, which replaces fixed-length segmentation with a flexible time-frequency resolution to capture localized fluctuations, enabling simultaneous analysis in the time-frequency domain.
Drawing inspiration from the Heisenberg-Gabor uncertainty principle, the graph BiMap layer strikes a balance by capturing signal information with high spectral resolution and low temporal resolution for low frequencies while adopting the opposite for high frequencies.
Specifically, to create a time-frequency graph within the graph BiMap layer, each spatial covariance matrix obtained from an EEG segment serves as a graph node.
The graph captures the local topology using a novel \emph{time-evolution} method.
The edge weight determines the similarity between neighboring nodes, computed using a Gaussian kernel applied to the Riemannian distance between two nodes on SPD manifolds.
Consequently, the classification problem of MI-EEG is transformed into a graph classification problem involving a time-frequency distribution.
Table~\ref{tab:tensor-graph} compares two geometric classifiers for MI-EEG classification.

Due to the limited number of trials available for training, the subject-specific scenario poses a big challenge for MI-EEG classifiers. Hence, the proposed geometric classifier will be evaluated on five publicly available MI-EEG datasets with two classic subject-specific scenarios: the 10-fold cross-validation and the holdout validation.
The rest of the paper is structured as follows: 
Section~\ref{preliminary} delves into the study history of MI-EEG classifiers and introduces SPD manifolds, graph convolutional networks, and existing neural networks on SPD manifolds. 
Section~\ref{methodology} outlines the architecture of Graph-CSPNet. 
The performance of Graph-CSPNet is then compared in a comprehensive set of experiments in Section~\ref{experiments}, and several key issues are discussed in detail in Section~\ref{sec:discussion}.

\section{Preliminary}\label{preliminary}

In the early stages of MI-EEG classification, feature extraction, and selection were commonly used. This involved hand-crafted features based on time-domain and frequency-domain statistics, such as band power, magnitude-squared coherence, and phase-locking value~\cite{brodu2011comparative,krusienski2012value,lachaux1999measuring}. Traditional machine learning approaches, including support vector machines and linear discriminant analysis, were then applied for classification.

The dominant perspective in MI-EEG classification has been the analysis in the time domain. Most methods utilize second-order statistics of EEG signals, specifically the spatial covariance matrix. The common spatial pattern (CSP) methods are widely used spatial filtering techniques. These methods aim to obtain optimal spatial features by maximizing the variance of one class while minimizing the variance of the other~\cite{muller1999designing,ramoser2000optimal,tomioka2006logistic,blankertz2007optimizing}.

This results in the following time-domain paradigm for EEG signal analysis.
Let $X \in \mathbb{R}^{n_C \times n_T}$ be a segment of EEG signals, where $n_C$ is the number of EEG channels, and $n_T$ is the number of sampled points on an epoch duration, and segment $X$ is assumed band-pass filtered, centered, and scaled. 
The linear classifier that predicts the label for segment $X$ can be written as follows,
\begin{align*}
f\big(X; \{w_i, \beta_i\}_{i=1}^N\big) = \sum_{i=1}^N \beta_i \log{(w_i X X^\top w_i^\top)} + \beta_0,
\end{align*} 
where $N$ is the number of spatial filters, $\{w_i\}_{i=1}^N \in \mathbb{R}^{n_C}$ are spatial filters and $\{\beta_i\}_{i=1}^N \in \mathbb{R}$ are biases. 
Significant attempts have been made to identify more versatile and practical spatial filters for BCI scenarios~\cite{lotte2018review}.

Spatial covariance matrices possess a substantial amount of discriminative information. This includes the variance captured in the on-diagonal entries and the coherence between adjacent channels recorded in the off-diagonal entries~\cite{yger2013review}. This information is closely related to rhythmic oscillations, which can be quantified through frequency analysis techniques like power spectral density and magnitude-squared coherence to track the ERD/ERS effect in motor imagery tasks~\cite{pfurtscheller201213}. By focusing solely on the sensorimotor cortices and utilizing spectral power to assess zero-phase synchrony, classification outcomes are comparable, if not superior, to those achieved using combinations of other spectral features, such as coherence and phase-locking value. This implies that the information embedded within spatial covariance matrices is adequate for MI-EEG classification~\cite{krusienski2012value}.

In 2011, Barachant first utilized SPD manifolds for characterizing EEG spatial covariance matrices in a Riemannian-based classification framework~\cite{barachant2011multiclass}.
The discipline of SPD manifolds is ingrained with the affine invariant Riemannian metric $g^{AIRM}$. 
This approach referred to as the Riemannian-based approach, has since garnered increasing interest in the BCI community, with subsequent studies focusing on optimizing its framework~\cite{barachant2011multiclass,yger2013review,congedo2017riemannian} and applications to various BCI tasks~\cite{barachant2011channel,sabbagh2019manifold,larzabal2021riemannian}. 
The success of Riemannian-based approaches can be attributed to their ability to capture the discriminative power between task classes by utilizing Riemannian distances.
The approach theoretically states in~\cite{barachant2010common} that the Riemannian distance between class-related covariance matrices $S^{+}$ and $S^{-}$ can be expressed as follows:
\begin{align}\label{formula_riemannian-based}
d_{g^{AIRM}} (S^{+}, S^{-}) = \sqrt{\sum_{i=1}^{n_C} \log^2 (\frac{\lambda_i}{1-\lambda_i}}),
\end{align}
where each $\lambda_i \in (0,1)$ is the eigenvalue of Equation~\ref{Eq_CSP} in Section~\ref{sec:CSP}.
This implies that a more significant Riemannian distance between two EEG spatial covariance matrices indicates a greater discriminative power, as their eigenvalues are more likely to be close to 1.

\begin{table*}[!ht]
\centering 
  \begin{threeparttable}[b]
\caption{MI-EEG classifiers using EEG spatial covariance matrices.~\label{tab:treatments}}
\begin{tabular}{ l | l | l}
\toprule
Category of Approach & Initiation Time & Principle\\ 
\midrule
Common Spatial Patterns (\cite{muller1999designing,ramoser2000optimal,tomioka2006logistic,blankertz2007optimizing}, etc.). & 1999~\cite{muller1999designing} &Search for a linear projection direction that maximizes the discriminative\\
&& information between classes through solving Eigenvalue Problem~\ref{Eq_CSP}.\\ 
Riemannian-Based Approaches (\cite{barachant2011multiclass,yger2013review,congedo2017riemannian}, etc.). & 2011~\cite{barachant2011multiclass} & Compare the quantity (i.e., Riemannian distance) consisting of the entire spectrum\\
&&information of Eigenvalue Problem~\ref{Eq_CSP}, instead of directly solving it.\\
Geometric Classifiers (\cite{ju2022tensor,ju2022deep,ju2023score,kobler2022spd,pan2022matt}, etc.). & 2022\tnote{*}~\cite{ju2022tensor} & Search for a nonlinear projection direction that maximizes the discriminative\\
&& information between classes using gradient descent.\\
\bottomrule
\end{tabular}
  \begin{tablenotes}
    \item[*] {\footnotesize Please note that while the initial application of second-order neural networks for processing EEG spatial covariance matrices was proposed in 2020~\cite{ju2020federated}, the systematic introduction of geometric classifiers in BCIs began with the publication of Tensor-CSPNet.}
  \end{tablenotes}
   \end{threeparttable}
\end{table*}

The exploration of SPD manifolds traces its roots back to the 1970s, springing from the classical task of enhancing geodetic networks~\cite{grafarend1974optimization}. 
Since then, it has been assimilated into a multitude of disciplines, including but not limited to geometric statistics, signal processing, computer vision, and robotics~\cite{pennec2008statistical,pennec2020manifold, jaquier2017gaussian,calinon2020gaussians,arnaudon2013riemannian,said2017riemannian,tuzel2006region,tabia2014covariance}. 
The conceptualization of spatial covariance matrices as points on SPD manifolds represents a significant advancement in the field, going beyond its purely formalistic nature. This abstraction acknowledges the positive definite constraints and infuses new vitality into research by incorporating technical tools from theoretical mathematics and physics. For instance, the abstract formulation of SPD manifolds provides a natural framework for interpolating diffusion tensors and introduces a novel measure of anisotropy in diffusion tensor magnetic resonance imaging~\cite{fletcher2007riemannian}. Moreover, applying optimal transport theory on SPD manifolds establishes a solid mathematical foundation for evaluating the variability between calibration-feedback phases in a BCI system~\cite{ju2022deep,bonet2023sliced}. 
Recently, there has been a gradual shift in the trend of MI-EEG classifiers towards the adoption of deep learning-based approaches that operate in the time domain, allowing for the automatic extraction of features from the original EEG signals.
In response to this trend, the first novel geometric deep learning-based approach on SPD manifolds, known as Tensor-CSPNet, has been proposed to leverage the \emph{time-space-frequency} principle in capturing discriminative information from EEG spatial covariance matrices. This approach has demonstrated performance levels close to optimality, comparable to those achieved by convolutional neural networks-based approaches~\cite{schirrmeister2017deep,sakhavi2018learning,lawhern2018eegnet,mane2021fbcnet}.
This approach is referred to as \emph{Geometric Classifiers} in this paper. 
The success of this approach is attributed to the BiMap layer introduced in Section~\ref{sec:NN_SPD}, which maintains the Riemannian distances consistently during the learning process, as indicated by the affine invariance property of $g^{AIRM}$ as follows, 
\begin{align}\label{formula_AIRM}
d_{g^{AIRM}}(S_1, S_2) = d_{g^{AIRM}}(WS_1W^\top, WS_2W^\top),
\end{align}
where $S_1, S_2 \in \mathcal{S}_{++}$ and weight matrix $W$ of the bi-map transformation belongs to the general linear group.
This preservation of Riemannian distances ensures the maintenance of discriminative power between task classes.
Numerous recent studies have showcased the consistency and potential of this pathway in the realm of MI-EEG classification tasks~\cite{ju2020federated,ju2023score,kobler2022spd,pan2022matt}.

In summary, Table~\ref{tab:treatments} presents the second-order statistics (EEG spatial covariance matrices) approaches for MI-EEG classifiers, with a more detailed discussion of time-domain statistics and frequency-domain statistics found in Section~\ref{TD_statistics}.
In the latter half of this section, we will briefly overview the computational techniques related to the proposed method.

\subsection{SPD Manifolds with Affine Invariant Riemannian Metric}
The spatial covariance matrice of an EEG segment $X\in \mathbb{R}^{n_C\times n_T}$ is denoted as $S = XX^\top \in \mathcal{S}_{++}$, where $\mathcal{S}_{++}:=\{S \in \mathbb{R}^{n_C \times n_C}: S = S^\top \text{ and } x^\top S x > 0, \forall x \in \mathbb{R}/\{0\} \}$ is the space of real $n_C\times n_C$ symmetric and positive definite matrices. 
Equipped with the affine invariant Riemannian metric $g^{AIRM}_P(v, w):= \langle \, P^{-\frac{1}{2}} v P^{-\frac{1}{2}}, P^{-\frac{1}{2}} w P^{-\frac{1}{2}} \, \rangle$ for each $v$ and $w$ on tangent space $\mathcal{T}_{P} \mathcal{S}_{++}$, the space of spatial covariance matrices becomes a Hadamard manifold, which refers to a Riemannian manifold that is complete and simply connected and has an everywhere non-positive sectional curvature. 
The Riemannian distance between two matrices $S_1$ and $S_2$ is given by $d_{g^{AIRM}} (S_1, S_2) = ||\log(S_i^{-\frac{1}{2}} S_j S_i^{-\frac{1}{2}})||_{\mathcal{F}}$, where $\mathcal{F}$ is the Frobenius norm.

\subsection{Common Spatial Pattern}\label{sec:CSP}
Let $S^{+}$ and $S^{-} \in \mathbb{R}^{n_C \times n_C}$ be the class-related estimates of spatial covariance matrices of segments in a two-class MI-EEG paradigm using the arithmetic mean, i.e., $S^+= \frac{1}{|\mathcal{I}_+|} \sum_{l \in \mathcal{I}_+ } X_l {X_l}^\top$ and $S^-= \frac{1}{|\mathcal{I}_-|} \sum_{l \in \mathcal{I}_- } X_l {X_l}^\top$, where $\mathcal{I}_+$ and $\mathcal{I}_-$ are sets of segments in two classes respectively. 
The CSP method is given by a simultaneous diagonalization of $S^{+}$ and $S^{-}$ as follows,
\[
W S^{+/-} W^\top = \Lambda^{+/-},
\]
where diagonal matrices $\Lambda^{+}$ and $\Lambda^{-} \in \mathbb{R}^{{n_C} \times {n_C}}$ hold an identity constraint $\Lambda^{+} + \Lambda^{-} = I_{n_C}$.
Each column vector $w \in \text{Column}(W)$ is called a spatial filter.
It is equivalent to solving a generalized eigenvalue problem as follows, 
\begin{align}\label{Eq_CSP}
(S^+ + S^-)^{-1} S^{+/-} w = \lambda^{+/-} w.
\end{align}
A set of discriminative powers consists of $m$ spatial filters $\{w_j\}_{j = 1}^m$ from both ends of the spectrum, i.e., $z_j = \log{(w_j XX^\top w_j^\top)}$, where $1 \leq m \leq {n_C}$.
%

\subsection{Graph Convolutional Networks}
Given an undirected graph $\mathcal{G}(\mathcal{V}, \mathcal{E})$ with $N$ nodes $\{v_i\}_{i=1}^N \in \mathcal{V}$ and edges $ \{(v_i, v_j)\}_{i\neq j} \in \mathcal{E}$.
Weighted adjacency matrix $A \in \mathbb{R}^{N \times N}$ of graph $\mathcal{G}$ records non-negative edge weights $\{a_{ij} \geq 0\}_{1\leq i,j\leq N}$.

The architecture of a multi-layer graph convolutional network model is given by the following layer-wise propagation rule:
\begin{align*}
H^{(l+1)} \gets \sigma \Big( (\bar{D}^{-\frac{1}{2}} \bar{A} \bar{D}^{-\frac{1}{2}}) H^{(l)} W^{(l)}\Big),
\end{align*}
where adjacency matrix with added self-connections $\bar{A} := A + I_N$, diagonal matrix $\bar{D}_{ii} := \sum_{j=1}^N \bar{A}_{ij}$, $\sigma$ is an activation function, and $W^{(l)}$ is a trainable weight matrix. $H^{(l)}$ is the matrix of activations in the $l$-th layer, and $H^{(0)}$ is the input data.
Theoretical analysis suggests that graph convolutional networks provide a localized first-order approximation of spectral graph convolution~\cite{kipf2017semi}.

\subsection{Neural Networks on SPD manifolds}\label{sec:NN_SPD}
The SPD matrix-valued network architecture in this section is derived from the following two studies.
\subsubsection{SPDNet~\cite{huang2017riemannian}} SPDNet introduces network architectures that operate on SPD matrices, incorporating the following layers:
\begin{itemize}
\item BiMap: This layer performs the bi-map transformation $WSW^T$ on spatial covariance matrix $S = U \Sigma U^T$.
The transformation matrix $W$ is typically required to have a full-row rank.
\item ReEig: This layer performs $U \max(\epsilon I, \Sigma) U^T$, where $\epsilon$ is a rectification threshold, and $I$ denotes the identity matrix.
\item LOG: This layer maps spatial covariance matrix $S$ onto its tangent space at identity matrix $I$ using $U\log(\Sigma)U^T$.
\end{itemize}
These layers are designed to ensure that the properties of symmetric and positive definite are maintained throughout the learning process.

\subsubsection{Riemannian Batch Normalization~\cite{brooks2019riemannian}}
This layer is designed as a type of SPD matrix-valued network architecture that utilizes parallel transport for batch centering and biasing around the Riemannian barycenter. 
Typically, parallel transport serves as a means to connect local information between various points because manifolds are the local concept, where basic operations that are simple in Euclidean spaces become intricate, leading to a distinct calculus on manifolds. 
This layer has been demonstrated to improve classification performance on EEGs, making it a valuable addition to the existing network architecture on SPD manifolds~\cite{ju2022tensor,kobler2022spd}.

\begin{figure*}[!h]
  \centering
  \includegraphics[width=0.85\linewidth]{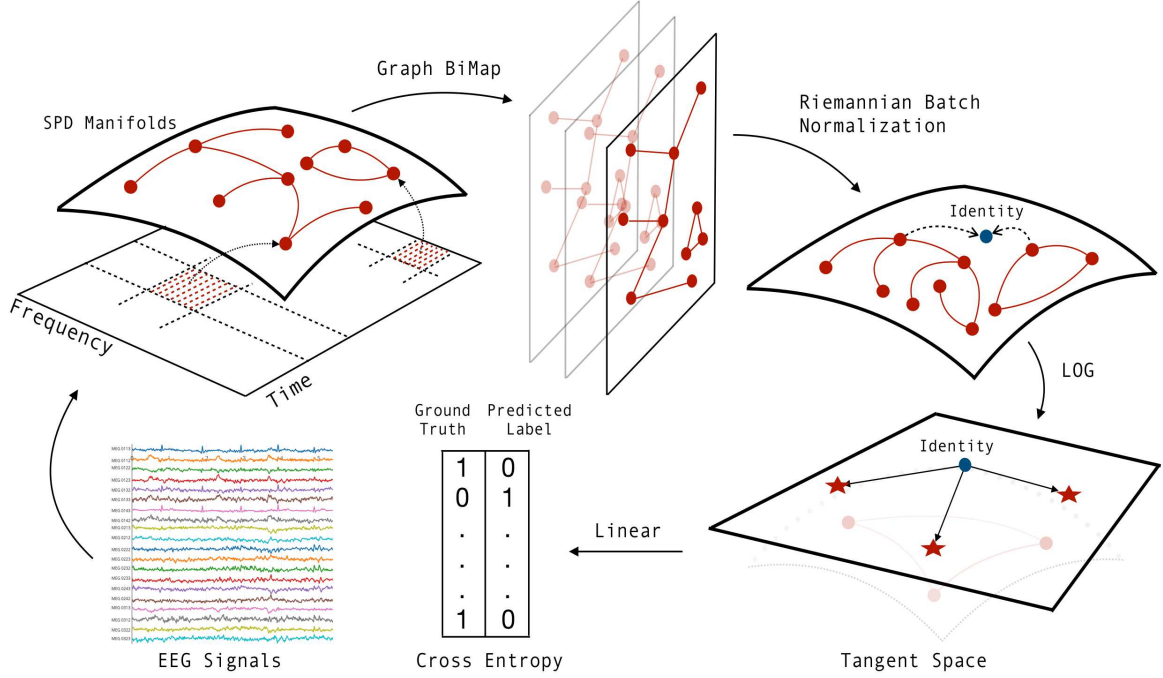}
  \caption{Illustration of Architecture of Graph-CSPNet: The EEG signal is divided into multiple segments in the time-frequency domain. The spatial covariance matrices obtained from these segments serve as the vertices of the time-frequency graph, which is constructed using a novel nonparametric statistical approach. Once the time-frequency graph is built, we employ an SPD matrix-valued graph convolutional network, consisting of the \emph{Graph BiMap} layer and the \emph{Riemannian Batch Normalization} layer, to extract relevant classification information while maintaining discriminative power between different task classes. Subsequently, we apply a logarithmic mapping (\emph{LOG}) layer to transform the SPD matrices onto the tangent space, which are then fed into the cross-entropy loss function for further computation.
\label{architecture}
} 
\end{figure*}

\section{Methodology}\label{methodology}

This section presents a novel geometric deep learning approach, Graph-CSPNet, for MI-EEG classification. The architecture of Graph-CSPNet is depicted in Figure~\ref{architecture}, and a table containing the parameters for each layer can be found in Section~\ref{appendix:architecture}.

\subsection{Time-Frequency Distribution}
A novel time-frequency distribution consisting of SPD matrices derived from EEG segments is constructed in the following way: given a segmentation plan on the time-frequency domain with or without overlapping $\{ \Delta t_i \times \Delta f_i\}_{i\in \mathcal{I}}$, the time-frequency distribution consists of EEG spatial covariance matrices $\{S(\Delta t_i \times \Delta f_i)\}_{i\in \mathcal{I}}$, where $S(\Delta t_i \times \Delta f_i):=\bar{X}_i {\bar{X}_i}^{\top}$ is the covariance matrix of band-pass filtered EEG signal $\bar{X}_i \in \mathbb{R}^{n_C \times \Delta t_i}$ within bandwidth $\Delta f_i$. 
The determination of a segmentation plan $\{ \Delta t_i \times \Delta f_i\}_{i\in \mathcal{I}}$ relies on changes in ongoing EEG activity, as evidenced by the appearance of the ERD/ERS effect that is induced by cognitive and motor processing. 
The ERD effect is characterized by a localized decrease in amplitude, while the ERS effect is marked by an increase in rhythmic activity amplitude. 
Both of these effects are highly specific to the frequency band of the event.
It is essential to consider the frequency discretization $\Delta f$ when working with traditional frequency bands, including $\delta$ (< 4 Hz), $\theta$ ($4 \sim 7$ Hz), $\mu$ ($8 \sim 13$ Hz), $\beta$ ($14 \sim 30$ Hz), and $\gamma$ (> 30 Hz) activity, which aligns closely with neurophysiological mechanisms. 
Due to the subject/user-specific nature of event-related discrimination, the discretization of frequency bands and time intervals can vary depending on the experiment.

\begin{remark}
To provide a concrete example, in the case of hand movement imagination, a time resolution of $\Delta t =  125$ ms is applicable to effectively capture the occurrence of the ERD/ERS effects that are specific to the $\mu$ and $\beta$ frequency bands, which offer the most effective discrimination. Precise selection of the time resolution is essential to ensure that the neural activity corresponding to these frequency bands is accurately detected and properly analyzed~\cite{pfurtscheller1997eeg,pfurtscheller2001functional}.
\end{remark}

\subsection{Time-Frequency Graph}\label{sec:graph}
To characterize the time-frequency distribution mentioned above, we construct a time-frequency graph denoted as $\mathcal{G}(\mathcal{V}, \mathcal{E})$. The vertices of this graph represent EEG spatial covariance matrices in the time-frequency distribution.
The edges of the time-frequency graph are created using a nonparametric statistical approach called the $\epsilon$-neighborhoods approach. This approach involves connecting two vertices, $S_i$ and $S_j$, with an edge if $d_{g^{AIRM}}(S_i, S_j)^2 < \epsilon$.

Expressly, we assume that brain activities generate a time evolution effect on the power spectrum of the EEG signals along the time axis.
To capture this effect, we employ a modified $\epsilon$-neighborhoods approach that establishes adjacency between two vertices, $S_i = S(\Delta t_i \times \Delta f_i)$ and $S_j = S(\Delta t_j \times \Delta f_j)$, if they fall within a box $\mathcal{B}_{\epsilon_1, \epsilon_2}$ in the time-frequency domain.
This box is defined by a box width $\epsilon_1\geq0$ on the time axis and a box height $\epsilon_2\geq0$ on the frequency axis. Therefore, $S_i$ and $S_j \in \mathcal{B}_{\epsilon_1, \epsilon_2}$ must satisfy the following conditions:
\begin{align*}
0 \leq \text{mid}(\Delta t_i)-\text{mid}(\Delta t_j) &\leq \epsilon_1; \\
|\text{mid}(\Delta f_i)-\text{mid}(\Delta f_j) | &\leq \epsilon_2,
\end{align*}
where "mid" represents the midpoint value of an interval. Note that the time evolution in this study occurs in the forward time flow direction, making it unnecessary to take the absolute value of the difference between the midpoints of $\Delta t$ in the above formula.

The adjacency matrix $A$ of $\mathcal{G}(\mathcal{V}, \mathcal{E})$ is used to store the similarities between pairs of vertices and is defined as follows,
\[
A = \left
\{ 
\begin{array}{rl}
e^{-d_{g^{AIRM}}^2(S_i, S_j)/t}, &  \mbox{if} \hspace*{0.4em} S_i \text{ and } S_j \text{ are adjacent}; \\ 
&\\
0, & \mbox{others.}  
\end{array}
\right.
\]
where preset Gaussian kernel width $t > 0$. 
These similarities are computed using the radial basis function kernel, which considers the Riemannian distance between adjacent vertices. If two vertices are not adjacent, the similarity is set to zero. In other words, if $S_i$ and $S_j$ are non-adjacent, the similarity value is assigned as zero.
The pseudocode for generating adjacency matrix $A$ refers to Section~\ref{sec:generation_A}.

This approach is referred to as the Local Graph Topology (LGT) method, as it effectively captures the local relationships within the time-frequency distribution.
By utilizing the LGT method, we gain a more comprehensive understanding of the intricate connections between EEG spatial covariance matrices within the time-frequency distribution.

\begin{remark}
1). Edge Weights: The weight of the edge in the time-frequency graph is calculated using a radial basis function kernel with Riemannian distance, as the discriminative power between task classes is closely related to the Riemannian distance.

2). Connected Components: According to the LGT approach, each frequency component ($\theta, \mu, \beta$, and $\gamma$ bands) forms a connected component in the adjacency matrix $A$ of the time-frequency graph, as shown in Table~\ref{tab:non-uniform}, 
\begin{equation*}
A = 
\begin{pmatrix}
A_{\theta} &  &  & \\
 & A_{\mu} &  &  \\
&& A_{\beta}  &  \\
& & & A_{\gamma} \\
\end{pmatrix}.
\end{equation*}
\end{remark}

\begin{table*}[!b]
\caption{Brief introductions to the five selected datasets, with the original settings specified in parentheses. The textual descriptions in the following text have been modified from the original settings to suit our purposes better.
~\label{tab:dataset_description}
}
\centering 
\begin{tabular}{l  l  l  l  l  l  l  l  l}
\toprule
 Dataset & Subjects & Channels  & Classes  & Trials/Session & Length & Imagery Period & Sampling Rate & Sessions  \\
\midrule
KU & 54 & 20 (62) & 2: left/right hand & 200 &  2.5 s & 1 to 3.5 s & 1000 Hz  & 2 \\
Cho2017 & 49 (52) & 20 (64) & 2: left/right hand & 200 & 3 s & 3 to 6 s & 512 Hz & 1 \\
BNCI2014001 & 9 & 22 & 4: left/right hand,feet,tongue & 288 &  4 s & 2 to 6 s & 250 Hz  & 2 \\
BNCI2014002 & 14 & 15 & 2: right hand,feet & 160 & 5 s & 3 to 8 s & 512 Hz & 1 \\
BNCI2015001 & 12 & 13 & 2: right hand,feet & 200 &  5 s & 3 to 8 s & 512 Hz  & 2 (2 or 3) \\
\bottomrule
\end{tabular}
\end{table*}

\subsection{Graph BiMap Layer}
This section introduces an SPD matrix-valued graph neural network that aims to extract discriminative information from the time-frequency graph.
To address this, we construct the layer-wise propagation rule for this SPD matrix-valued graph neural network as follows:
\begin{equation*}
H^{(l+1)} \gets \text{RBN} \Big( \text{ReEig} \big( W^{(l)} (\bar{D}^{-1} \bar{A}^{(l)}) H^{(l)} {W^{(l)}}^\top \big) \Big), 
\end{equation*}
where $\bar{A}^{(l)} := A^{(l)} + I_N$, $\bar{D}_{ii} := \sum_j \bar{A}^{(l)}_{ij}$, and $W^{(l)}$ is a trainable transformation matrix with the full-row rank. 
The RBN (Riemannian batch normalization) and ReEig layers are introduced in Section~\ref{sec:NN_SPD}. 
$H^{(l)} \in \mathbb{R}^{|\mathcal{V}| \times n_C^2}$ is a node function in the $l^{th}$ layer. 
In particular, $H^{(0)}:=(S^1, \cdots, S^N)$, $A^{(0)}$ is the time-frequency graph, and $\bar{A}^{(l)} := I_N$, for $l \geq 1$.

The proposed layer-wise propagation rule is structurally similar to the classical rule in graph convolutional network architectures. The original linear mapping is replaced with the BiMap mapping, which enables operations on SPD matrices. BiMap ensures that Riemannian distances are consistently preserved during the learning process while maintaining symmetric positive definiteness. This is crucial for MI-EEG classification tasks. 
Furthermore, instead of using $\bar{D}^{-\frac{1}{2}} \bar{A}^{(l)} \bar{D}^{-\frac{1}{2}}$, we employ row-normalized $\bar{D}^{-1} \bar{A}^{(l)}$ because we aim to perform normalization only in the time direction. Additionally, we utilize the non-linear operator ReEig, which drops the smallest eigenvalues, as described in Section~\ref{appen:projection}, to prevent matrix degeneracy. 

However, the proposed rule still differs from the classical one. In the following, we will discuss these differences in detail, specifically focusing on their relevance to our application. A perturbation analysis for the spectrum change induced by transformation $\bar{D}^{-1} \bar{A}^{(l)}$ is presented as follows,
\begin{theorem}~\label{thm:perturbation}
Given a time-frequency graph $\mathcal{G}(\mathcal{V}, \mathcal{E})$. For the $l^{th}$ graph BiMap layer, let $A^{(l)}$ be its $|\mathcal{V}|\times |\mathcal{V}|$ adjacency matrix.
For $i \in \{1, ..., |\mathcal{V}|\}$, the perturbated spatial covariance matrix $\bar{S}_i \in \mathbb{R}^{n_C \times n_C}$ on ($\mathcal{S}_{++}, g^{AIRM}$) is wtitten in to the original spatial covariance matrix $S_i$ and a graph-based perturbation term, i.e., $\bar{S}_i = S_i + \Delta S_i$, where $\Delta S_i := A^{(l)}[i, :] (S_1, \cdots, S_{|\mathcal{V}|})^{\top}$.
Then, the spectrum of the spatial covariance matrix has a perturbation ratio as follows, 
\begin{align}\label{equ:Thm1}
\big(1-N_i C_i \big) \leq \frac{\lambda (\bar{S}_i)}{\lambda (S_i)} \leq \big(1+N_i C_i\big), 
\end{align}
where $C_i = \max_{(i, j)\in \mathcal{E}} \{\exp{(\lambda_{ij})}\}$, $N_i =  \Big| \{j \big| (i, j)\in\mathcal{E}\} \Big|$, and $\lambda_{ij}$ is largest eigenvalue of $\log (S_i^{-\frac{1}{2}} S_j S_i^{-\frac{1}{2}})$.
\end{theorem}

It is imperative that we initially establish the subsequent Lemma to prove Theorem~\ref{thm:perturbation}.

\begin{lemma}\label{l1}
Let $\{S_i\}_{i=1}^N$ be a set of SPD matrices. Then, any linear combination $\sum_{i=1}^N \alpha_i S_i$ is still SPD for $\alpha_i > 0$.
\end{lemma}

\begin{proof}
The symmetry is obtained by $(\sum_{i=1}^N \alpha_i S_i)^{\top} = \sum_{i=1}^N \alpha_i S_i^{\top} = \sum_{i=1}^N \alpha_i S_i$, and the positive definite is achieved by $v^{\top} (\sum_{i=1}^N \alpha_i S_i)  v =  \sum_{i=1}^N \alpha_i (v^{\top} S_i v) > 0$, for any $v \neq 0$. This concludes the proof. 
\end{proof}

\begin{proof}[Proof of Theorem~\ref{thm:perturbation}]
Consider the perturbation $\Delta S_i $ on each node $S_i$ ($1\leq i \leq |\mathcal{V}|$), which is taken as a linear combination of adjacent nodes as $S_i + \Delta S_i = S_i + \sum_{(i, j) \in \mathcal{E}} a_{ij} S_j$, where $a_{ij}:=e^{-d^2_{g^{AIRM}}(S_i, S_j)/t} \leq 1$ with a large $t$ is the similarity value for the edge between node $S_i$ and $S_j$ in adjacency matrix $A^{(l)}$. 
Each $S_i + \Delta S_i$ remains symmetric positive and definite according to Lemma~\ref{l1}.

The Riemannian distance between node $S_i$ and its adjacent node $S_j$ is $d_{g^{AIRM}}(S_i, S_j) := ||\log (S_i^{-\frac{1}{2}} S_j S_i^{-\frac{1}{2}})||_{\mathcal{F}}$. 
Notice that $|\log (S_i^{-\frac{1}{2}} S_j S_i^{-\frac{1}{2}})$ is a $n_C\times n_C$ symmetric matrix and bounded by an inequality according to the Rayleigh-Ritz theorem~\cite{demmel1997applied} as $-\lambda_{ij} I \preceq \log (S_i^{-\frac{1}{2}} S_j S_i^{-\frac{1}{2}}) \preceq \lambda_{ij} I$, where $\lambda_{ij}$ is the largest Rayleigh quotient of $ \log (S_i^{-\frac{1}{2}} S_j S_i^{-\frac{1}{2}})$.
Since $\exp({\lambda_{ij}})$ is the eigenvalue of $S_i^{-\frac{1}{2}} S_j S_i^{-\frac{1}{2}}$ according to the definition of exponential matrix and $\exp$ is monotone increasing, we have $-\exp({\lambda_{ij}}) I \preceq S_i^{-\frac{1}{2}} S_j S_i^{-\frac{1}{2}} \preceq \exp({\lambda_{ij}}) I$.

Note that the pertubation $S_i + \Delta S_i =  S_i + \sum_{(i, j) \in \mathcal{E}} a_{ij} S_j = S_i^{\frac{1}{2}} \Big(I + S_i^{-\frac{1}{2}}\big(\sum_{(i, j) \in \mathcal{E}} a_{ij} S_j \big) S_i^{-\frac{1}{2}}\Big)S_i^{\frac{1}{2}}$.
Then, we have $\big(1- N_i C_{i}\big) S_i \preceq S_i + \sum_{(i, j) \in \mathcal{E}} a_{ij} S_j \preceq \big(1+ N_i C_{i}\big) S_i$, where $N_i = |\{j| (i, j) \in \mathcal{E}\}|$ and $C_i = \max_{(i, j)\in \mathcal{E}} \{\exp{(\lambda_{ij})}\}$. 
Lastly, the spectrum perturbation ratio in Equation~\ref{equ:Thm1} can be obtained using Weyl's monotonicity theorem~\cite{horn1994topics}.
\end{proof}
Theorem~\ref{thm:perturbation} provides both a rough estimated upper and lower bounds for spectrum change ratio $\lambda (\bar{S}_i)/\lambda (S_i)$, for $1 \leq i \leq N$.
This ratio depends solely on the node degree and spectrum between vertices in the time-frequency graph.
As the depth of graph BiMap layers increases, the spectrum change ratio increases accordingly. 

To maintain the lowest possible spectrum change ratio, it is necessary to set the depth of the graph BiMap layer to one, i.e., a choice is to put $\bar{A}^{(l)}:= I_N$, for $l > 1$, following the first layer.
The visualization of the impact of spectrum changes, as stated in Theorem~\ref{thm:perturbation}, is discussed in detail in Section~\ref{sec:spectrum_distribution_shift}.

\subsection{LOG Layer}
After the graph BiMap layers, the resulting SPD matrices are mapped onto the tangent space using the logarithm mapping with the identity matrix as the base point, as explained in the LOG layer in Section~\ref{sec:NN_SPD}.
In particular, this logarithm mapping is well-defined and applicable to any SPD matrix with the same dimension.

\subsection{Loss Function}
For the sake of simplicity, the cross-entropy loss is adopted as the loss function for Graph-CSPNet. 

\begin{table*}[!t]
\caption{A non-overlapping and non-uniform segmentation plan for Graph-CSPNet. The table provided lists each frequency band's time-window length (seconds) as a distinct entity.
~\label{tab:non-uniform}
}
\centering 
\begin{tabular}{l  l  l  l  l  l  l  l  l  l}
\toprule
{} &   \multicolumn{1}{c}{$\theta$ band} &   \multicolumn{1}{|c}{$\mu$ band} &  \multicolumn{4}{|c}{$\beta$ band} &  \multicolumn{3}{|c}{$\gamma$ band} \\
\midrule
 {Dataset/Freq Band (Hz)} & $4 \sim 8$ & $8 \sim 12$ &  $12 \sim 16$  &  $16 \sim 20$ &  $20 \sim 24$ &  $24 \sim 28$  & $28 \sim 32$  & $32 \sim 36$ & $36 \sim 40$ \\
\midrule
KU/Cho2017 & 0.5 & 0.5 &  0.5 &  0.5 &  0.5 &  0.5  & 0.25   & 0.25 & 0.25  \\
BNCI2014001 & 0.25 & 0.25 &  0.25 &  0.25 &  0.25  &  0.25 & 0.125 & 0.125 & 0.125 \\
BNCI2014002/2015001 & 1 & 1  & 1  & 1  & 1  &1  & 0.5  & 0.5 & 0.5  \\
\bottomrule
\end{tabular}
\end{table*}

\section{Experiments}\label{experiments}

\subsection{Evaluation Datasets and Scenarios}\label{datasets}
As part of the evaluation, we employed five widely-used MI-EEG datasets. 
This section offers a concise overview of these datasets.
The data of these datasets cannot be accessed via the MOABB~\footnote{\, The MOABB package comprises a benchmark dataset for state-of-the-art decoding algorithms encompassing 12 open access datasets and over 250 subjects. This package is accessible at~\url{https://
github.com/NeuroTechX/moabb}.} or BNCI Horizon 2020~\footnote{\, The datasets within BNCI Horizon 2020 project is accessible at~\url{http://bnci-horizon-2020.eu/database/data-sets}.}.
All digital signals underwent filtering using Chebyshev Type II filters with 4 Hz intervals. The filters were specifically designed to have a maximum passband loss of 3 dB and a minimum stopband attenuation of 30 dB.
The selection criteria for the five evaluation datasets from MOABB and BNCI Horizon 2020 requires each training session to have more than 150 trial instances, with at least 70 trials for each category.
Table~\ref{tab:dataset_description} summarizes the five evaluation datasets. 

1). KU (also known as \emph{Lee2019\_MI}, from MOABB): It comprises EEG signals obtained from 54 subjects participating in a binary class EEG-MI task. The EEG signals were captured at 1,000 Hz using 62 electrodes. For evaluation, 20 electrodes in the motor cortex region were selected, i.e., FC-5/3/1/2/4/6, C-5/3/1/z/2/4/5, and CP-5/3/1/z/2/4/6. The dataset was divided into two sessions, S1 and S2, each composed of two phases: training and testing. Each phase consisted of 100 trials balanced between right and left-hand imagery tasks, resulting in 21,600 trials available for evaluation. The EEG signals were epoched from the first second to 3.5 seconds in reference to the stimulus onset, resulting in a total duration of 2.5 seconds. 

2). Cho2017 (from MOABB): It conducted a BCI experiment with 52 subjects involving motor imagery movement (MI movement) of the left and right hands. The EEG data were collected using 64 Ag/AgCl active electrodes, with a 64-channel montage based on the international 10-10 system to record the EEG signals at a sampling rate of 512 Hz. The recording commenced at second 0, following the prompt, and continued for 3 seconds until the conclusion of the cross period. For evaluation, 20 electrodes were selected from the motor cortex region, comprising FC-5/3/1/2/4/6, C-5/3/1/z/2/4/5, and CP-5/3/1/z/2/4/6. It is worth noting that the EEG data of subjects No.32, No.46, and No.49 were omitted, reducing the total number of subjects to 49.

3). BNCI2014001 (also known as BCIC-IV-2a, from BNCI Horizon 2020): Its studies involved 9 participants who executed a motor imagery task comprising four classes: left hand, right hand, feet, and tongue. The task was recorded using 22 Ag/AgCl electrodes and three EOG channels at a sampling rate of 250 Hz. The recorded data were filtered between 0.5 and 100 Hz, with a 50 Hz notch filter applied. The study was conducted in two sessions on separate days, with each participant completing 288 trials in total (six runs of 12 cue-based trials for each class). The EEG signals were epoched from the cue's onset at Second 2 to the end of the motor imagery period at Second 6, resulting in a total duration of 4 seconds.

4). BNCI2014002 (from BNCI Horizon 2020): 13 participants were tasked with carrying out sustained 5-second kinaesthetic MI of their right hand and feet, as directed by the cue. The session comprised eight runs consisting of 80 trials for each class, amounting to 160 trials per participant. EEG measurement was conducted using a biosignal amplifier and active Ag/AgCl electrodes with a sampling rate of 512 Hz. A total of 15 electrodes were placed. The EEGs were epoched from the cue onset at second 3 to the end of the motor imagery period at second 8, resulting in a total duration of 5 seconds. 

5). BNCI2015001(from BNCI Horizon 2020): It involved 12 subjects who performed a sustained right-hand versus both feet imagery task. Most of the study was composed of two sessions, conducted on consecutive days, with each participant completing 100 trials for each class, amounting to 200 trials per participant. Four subjects participated in a third session, but the evaluation did not include these data. The data were acquired at a sampling rate of 512 Hz, utilizing a bandpass filter ranging from 0.5 to 100 Hz and a notch filter at 50 Hz. The recording began at Second 3 with the prompt and continued until the end of the cross period at Second 8, with a total duration of 5 seconds.

We typically evaluate datasets containing two sessions, namely KU, BNCI2014001, and BNCI2015001, using two distinct scenarios. 
The first is known as 10-Fold Cross Validation (CV), whereby we divide each subject's data into ten equally sized class-balanced folds, with nine used for training and one used for testing, and this process is repeated ten times. 
The second scenario is called Holdout, where we employ the first session for training and the second for testing. 
Given that the EEG data for the two sessions is usually gathered on different days, there may be considerable variability. 
For the KU dataset, we utilize the first 100 trials of the second session for validation, with the remaining trials used for testing. 
In contrast, datasets featuring only one session, such as BNCI2014002 and Cho2017, are evaluated solely using the 10-fold CV scenario.

\subsection{Evaluated Segmentation Plans}
This study presents an example of non-overlapping and non-uniform segmentation plans, as shown in Table~\ref{tab:non-uniform}, to evaluate the proposed approach. 
The term "non-overlapping" indicates no overlap between time windows, while "non-uniform" implies that the size of the time windows varies within the frequency band. 
We considered the commonly used method of segmenting signals in the time domain into units of seconds, half-seconds, or quarter-seconds, which BCI researchers have widely adopted in the past years, but it is not based on extensive neurophysiological considerations.
In Table~\ref{tab:non-uniform}, there are 60 segments of EEG signals in the KU/BNCI2014002/BNCI2015001 datasets (i.e., 60 segments = 5 windows $\times$ 6 frequency bands + 10 windows $\times$ 3 frequency bands), 48 segments in the BNCI2014001 dataset (i.e., 48 segments = 4 windows $\times$ 6 frequency bands + 8 windows $\times$ 3 frequency bands), and 33 segments in the Cho2017 dataset (33 segments = 3 windows $\times$ 6 frequency bands + 5 windows $\times$ 3 frequency bands).

\begin{table*}[!t]
\caption{Average accuracies and corresponding standard deviations derived from subject-specific analyses of the KU, Cho2017, BNCI2014001, BNCI2014002, and BNCI2015001 dataset.
Each result in the table is expressed as the average accuracy and its corresponding standard deviation. Notably, the optimal outcome for each analysis is highlighted in boldface, thus providing an enhanced visual representation of the best-performing metrics.
}
\centering 
\begin{tabular}{l | l | llll}
\toprule
Dataset & Scenario & FBCSP & FBCNet & Tensor-CSPNet & Graph-CSPNet \\
\midrule
KU & 10-Fold CV (S1) \% & 64.33 (15.43) & \textbf{73.36} (13.71) & 73.28 (15.10) & 72.51 (15.31)\\
{}   & 10-Fold CV (S2) \% & 66.20 (16.29) & 73.68 (14.97) & 74.16 (14.50) &  \textbf{74.44} (15.52)\\
{}   & Holdout (S1 $\rightarrow$ S2) \% & 59.67 (14.32) & 67.74 (14.52) & 69.50 (15.15) &\textbf{69.69} (14.72)\\
\midrule
Cho2017 &10-Fold CV \% & 61.75 (13.26) & 65.34 (11.14) & 67.30 (12.94) & \textbf{67.51} (12.89) \\
\midrule
BNCI2014001 &10-Fold CV (T) \% &71.29 (16.20) & 75.48 (14.00) & 75.11 (12.68) & \textbf{77.55} (15.63) \\
{(BCIC-IV-2a)} &10-Fold CV (E) \% & 73.39 (15.55) & 77.16 (12.77) & 77.36 (15.27) &\textbf{78.82} (13.40)\\
{} & Holdout (T $\rightarrow$ E) \% & 66.13 (15.54) & 71.53 (14.86) & \textbf{73.61} (13.98) & 71.95 (13.36) \\
\midrule
BNCI2014002 &10-Fold CV \% & 76.07 (13.29) & 79.64 (12.77) & 80.58 (11.87) & \textbf{81.65} (11.74) \\
\midrule
BNCI2015001 &10-Fold CV (A) \% & 79.46 (14.16) & 82.62 (13.11) & 81.29 (14.78) & \textbf{84.62} (12.38) \\
{} & 10-Fold CV (B) \% & 81.96 (11.14) & 84.92 (10.30) & 85.29 (10.54) & \textbf{88.00} (7.87) \\
{} & Holdout (A $\rightarrow$ B) \% & 73.46 (14.09) & 74.50 (16.01) & 79.04 (14.67) & \textbf{79.75} (14.63) \\
\bottomrule
\end{tabular}
\label{tab:tab_acc}
\end{table*}

\subsection{Evaluation Baselines}
This study will compare the proposed approach with a range of baseline methods, including FBCSP, FBCNet, and Tensor-CSPNet. The three selected baselines are representative algorithms utilizing the ERD/ERS effect for the MI-EEG classification during an MI process through various technological periods.

\begin{enumerate}
\item FBCSP~\cite{ang2008filter}: This approach is a prominent method within the CSP family, which utilizes the CSP algorithm to extract sub-band scores from EEG signals and subsequently employs classification algorithms such as Support Vector Machines (SVMs) on the selected features. 
\item FBCNet~\cite{mane2021fbcnet}: This approach is a convolutional neural networks-based classifier designed to classify EEG signals by analyzing EEG patterns within the space-time-frequency domain. Specifically, FBCNet performs on sub-frequency bands of EEGs and is highly skilled at capturing complex EEG patterns. 
\item Tensor-CSPNet: This approach represents the first geometric deep learning approach to MI-EEG classification, capable of exploiting EEG patterns sequentially across the frequency, space, and time domains by leveraging existing deep neural networks on SPD manifolds.
\end{enumerate}
It is worth noting that comparison with other prevalent MI-EEG classifiers, such as the convolutional neural networks-based approaches and Riemannian-based approaches, is not deemed necessary as previous studies~\cite{ju2022tensor} have established that FBCNet and Tensor-CSPNet exhibit superior performance on the subject-specific scenarios of datasets KU and BNCI2014001 under evaluation.
Furthermore, while the proposed method is referred to as graph neural networks, it should be noted that it fundamentally differs from utilizing EEG device channels as graph nodes to construct a graph neural network in the context of MI-EEG classifiers. The method outlined in the text is the most suitable for comparison with the proposed method. Therefore, other graph-based MI-EEG classifiers will not be considered or adopted for this study.

\subsection{Configurations of Network Architecture}
For the deep learning methodology, the choice of neural network configuration plays a crucial role in the performance. 
Hence, we will present neural network configurations for all deep learning-based approaches used in the evaluation.

For baseline FBCNet, it has relatively few neural network hyperparameters that can be tuned. We used the same nine frequency sub-bands as FBCSP and set 16 parallel spatial convolution blocks for all datasets.
For the geometric classifiers, we have adopted simple yet effective network architectures for different scenarios as follows: 
\begin{itemize}
\item Tensor-CSPNet: The network architecture consists of a two-layer BiMap block (BiMap layer + ReEig layer + RBN layer), with input/output dimensions varying across different datasets. The frequency segmentation for all datasets is performed in 4 Hz bandwidth increments from 4 Hz to 40 Hz without any overlap. The time segmentation, however, varies for each dataset, as follows:
1). For the KU dataset, the BiMap block transforms the input dimension from 20 to 30 and then back to an output dimension of 20. The network employs three temporal segmentations: 0 to 1.5 seconds, 0.5 to 2 seconds, and 1 to 2.5 seconds.
2). In the Cho2017 dataset, the BiMap block converts the input dimension from 20 to 30 and then reverts it to an output dimension of 20. The network is divided into three temporal segments: 0 to 1 second, 1 to 2 seconds, and 2 to 3 seconds.
3). For the BNCI2014001 dataset, the BiMap block transforms the input dimension from 22 to 36 and then back to an output dimension of 22. This network has two temporal segmentations: 0 to 0.75 seconds and 0.25 to 1 second.
4). For the BNCI2014002 dataset, the BiMap block increases the input dimension from 15 to 30 and then reduces it back to an output dimension of 15. This network employs five temporal segmentations: 0 to 1 second, 1 to 2 seconds, 2 to 3 seconds, 3 to 4 seconds, and 4 to 5 seconds.
5). For the BNCI2015001 dataset, the BiMap block expands the input dimension from 13 to 30 and then reduces it to an output dimension of 13. This network utilizes five different temporal segmentations: 0 to 1 second, 1 to 2 seconds, 2 to 3 seconds, 3 to 4 seconds, and 4 to 5 seconds.

\item Graph-CSPNet: Graph-CSPNet shares the same neural network architecture as Tensor-CSPNet across all scenarios. However, it employs a different segmentation approach for frequency and time. The specific time-frequency segmentation plan for Graph-CSPNet is outlined in Table~\ref{tab:non-uniform}.
In addition, we typically set the forward time flow to half their maximum possible value, as discussed in Section~\ref{sec:hyperparameters}.
\end{itemize}

\subsection{Classification Performance}

The evaluation is conducted using subject-specific scenario settings. Two datasets, namely KU and BNCI2014001, have been reinitialized with cross-validation indices, employing different network architectures than those utilized in our previous study~\cite{ju2022tensor}. Consequently, the classification accuracies of each baseline may slightly vary from the earlier results within an acceptable range of approximately 1\%. The remaining three datasets are being utilized for the first time in our study. For each scenario, we select the best result among multiple runs.

The \emph{time-space-frequency} principle is fundamental to all three MI-EEG classification algorithms: FBCNet, Tensor-CSPNet, and Graph-CSPNet. Despite differences in network architectures, they all effectively utilize discriminative information from the time, space, and frequency domains.
The segmentation technique employed by Tensor-CSPNet and Graph-CSPNet contributes to their superior performance over FBCNet in the three holdout scenarios. This technique breaks down EEG signals into short, quasi-stationary intervals, thereby reducing distribution shifts.
Among the algorithms, Graph-CSPNet demonstrates a slight advantage in nine out of the eleven scenarios listed in Table~\ref{tab:tab_acc}. This is attributed to its refined time-frequency segmentation technique, which provides a more precise characterization of EEG signals.

\subsection{Hyperparameters in Graph-CSPNet}~\label{sec:hyperparameters}
We will investigate the influence of hyperparameters in the time-frequency graph, such as $\epsilon$ in the LGT method, on the performance of Graph-CSPNet. This analysis is depicted in Figure~\ref{fig:LGT_hyperparameter}. To streamline the evaluation process, the network configuration of Graph-CSPNet will be simplified to include a depth-1 graph BiMap layer with input and output dimensions set to 20.

\begin{enumerate}
\item The utilization of graph topology enables Graph-CSPNet to learn discriminative patterns from the time-frequency graph effectively. This is evident from the significant improvements observed in other time-frequency graph configurations compared to the "non-graph" case, i.e., Time Direction (0,0,0,0) in Figure~\ref{fig:LGT_hyperparameter}. The "non-graph" case refers to the similarity matrix $A^{(0)} := I_N$.

\item Among different time directions, Time Direction (2,2,2,4) yields the best performance, as demonstrated in Figure~\ref{fig:LGT_hyperparameter}. The corresponding transformed spectral distributions can be found in Figure~\ref{fig:spectrum} (c).
Additionally, Time Direction (2,2,2,4) showed statistical significance in the 10-fold CV experiments of the two groups. However, it did not exhibit significance in the holdout experiment.
Statistical significance was determined using the one-tailed Wilcoxon signed-rank test~\cite{demvsar2006statistical} with Bonferroni-Holm correction, with a significance level of $\alpha=0.05$.
Therefore, in all of our experiments, we generally set the forward time flow to the nearest integer values, which is the closest approximation to half-values.
\end{enumerate}

\begin{figure}[!h]
  \centering
  \includegraphics[width=\linewidth]{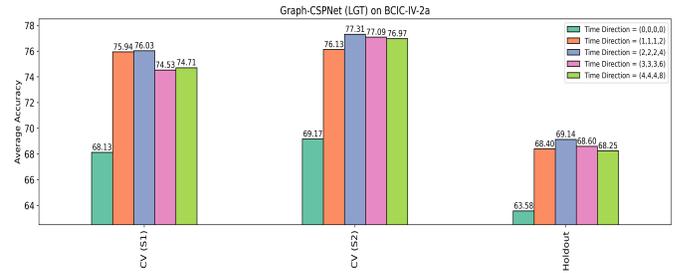}
  \caption{The classification performance of Graph-CSPNet with multiple time and frequency directions is evaluated on the BNCI2014001 (BCIC-IV-2a) dataset. The forward time direction encompasses the intervals $0 \leq x_{\theta}\leq 4$, $0 \leq x_{\mu} \leq 4$, $0 \leq x_{\beta} \leq 4$, and $0 \leq x_{\gamma} \leq 8$. The frequency direction is fixed at Time Direction $(1,1,4, 3)$.
\label{fig:LGT_hyperparameter}
} 
\end{figure}

\begin{figure}[!h]
  \centering
  \includegraphics[width=\linewidth]{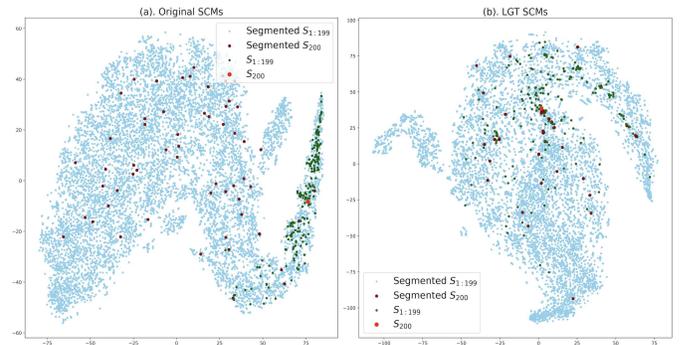}
  \caption{Spectrum Distribution Shift: A two-dimensional projection of Subject No. 1 in the KU dataset using t-SNE is depicted. Each session for the subject consists of a total of 200 trials. The bright red point in each subfigure represents the $200^{th}$ trial in that session, while the 199 green points represent the first 199 trials. These points have undergone dimensionality reduction from $20 \times 20$-dimensional EEG spatial covariance matrices using t-SNE. Furthermore, the 60 dark red points represent the spatial covariance matrices obtained from 60 EEG segments, as specified in the segmentation plan outlined in Table~\ref{tab:non-uniform}, from the last trial (represented by the bright red point). The 11,940 ($=199 \times 60$) blue points also correspond to the spatial covariance matrices derived from the remaining 199 trials.
\label{fig:Dist_LGT}
} 
\end{figure}

\begin{figure*}[!h]
\center

\begin{subfigure}{0.33\textwidth}
\includegraphics[width=\linewidth]{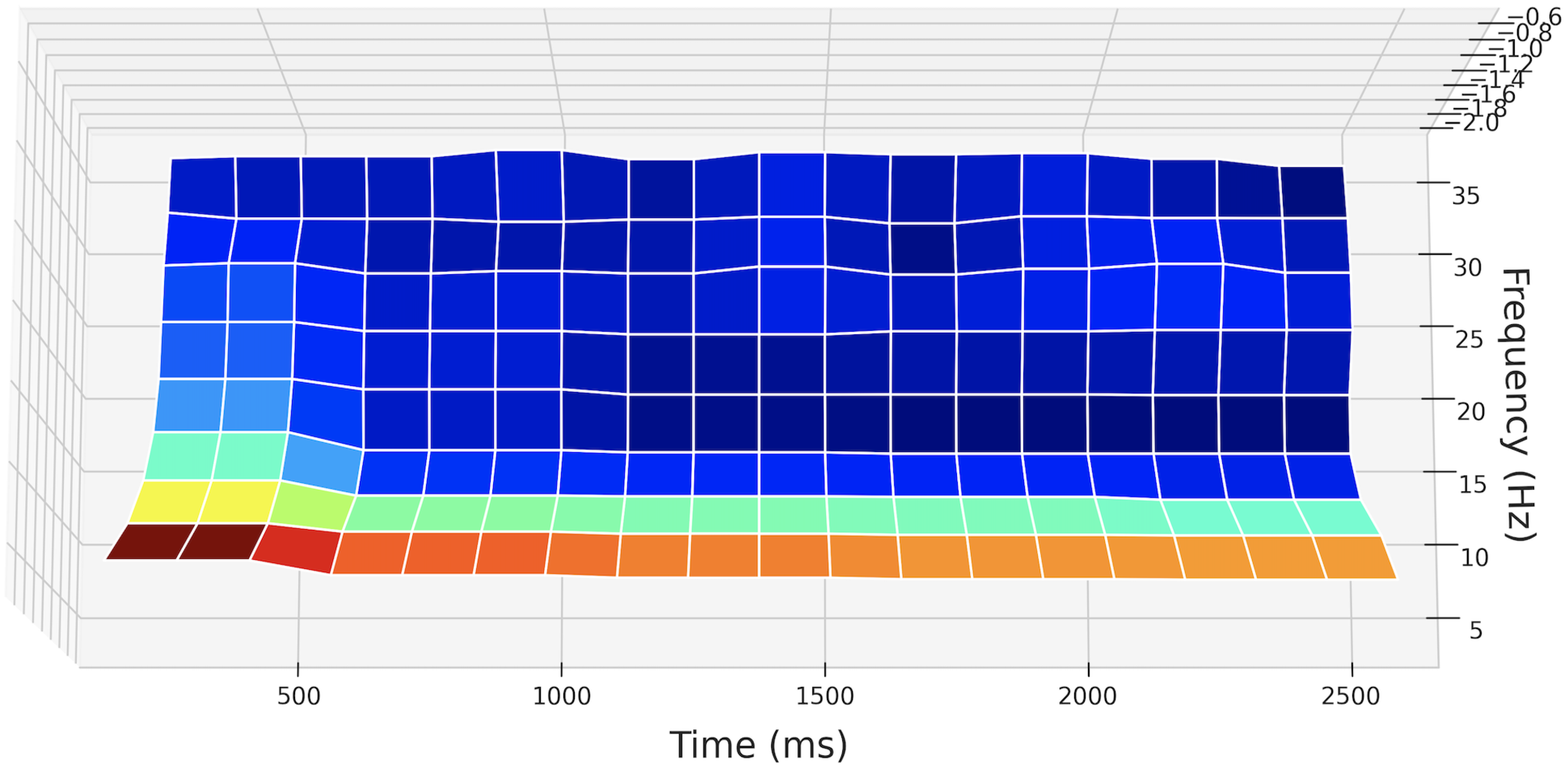}
\caption{Discrete Spectrogram}~\label{fig:a}
\end{subfigure}\hspace*{\fill}
\begin{subfigure}{0.33\textwidth}
\includegraphics[width=\linewidth]{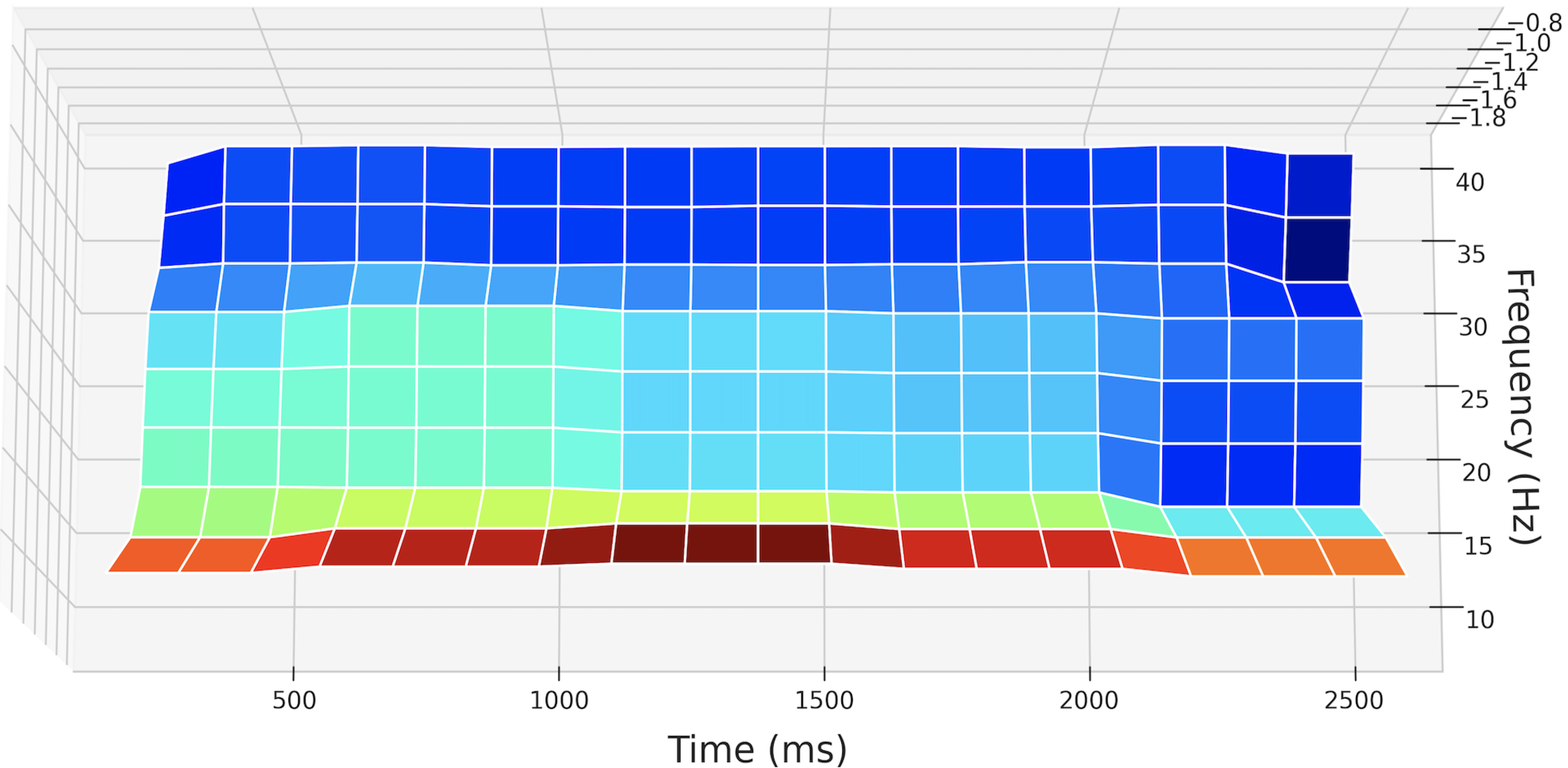}
\caption{Time Direction (1,1,1,2).}~\label{fig:b}
\end{subfigure}\hspace*{\fill}
\begin{subfigure}{0.33\textwidth}
\includegraphics[width=\linewidth]{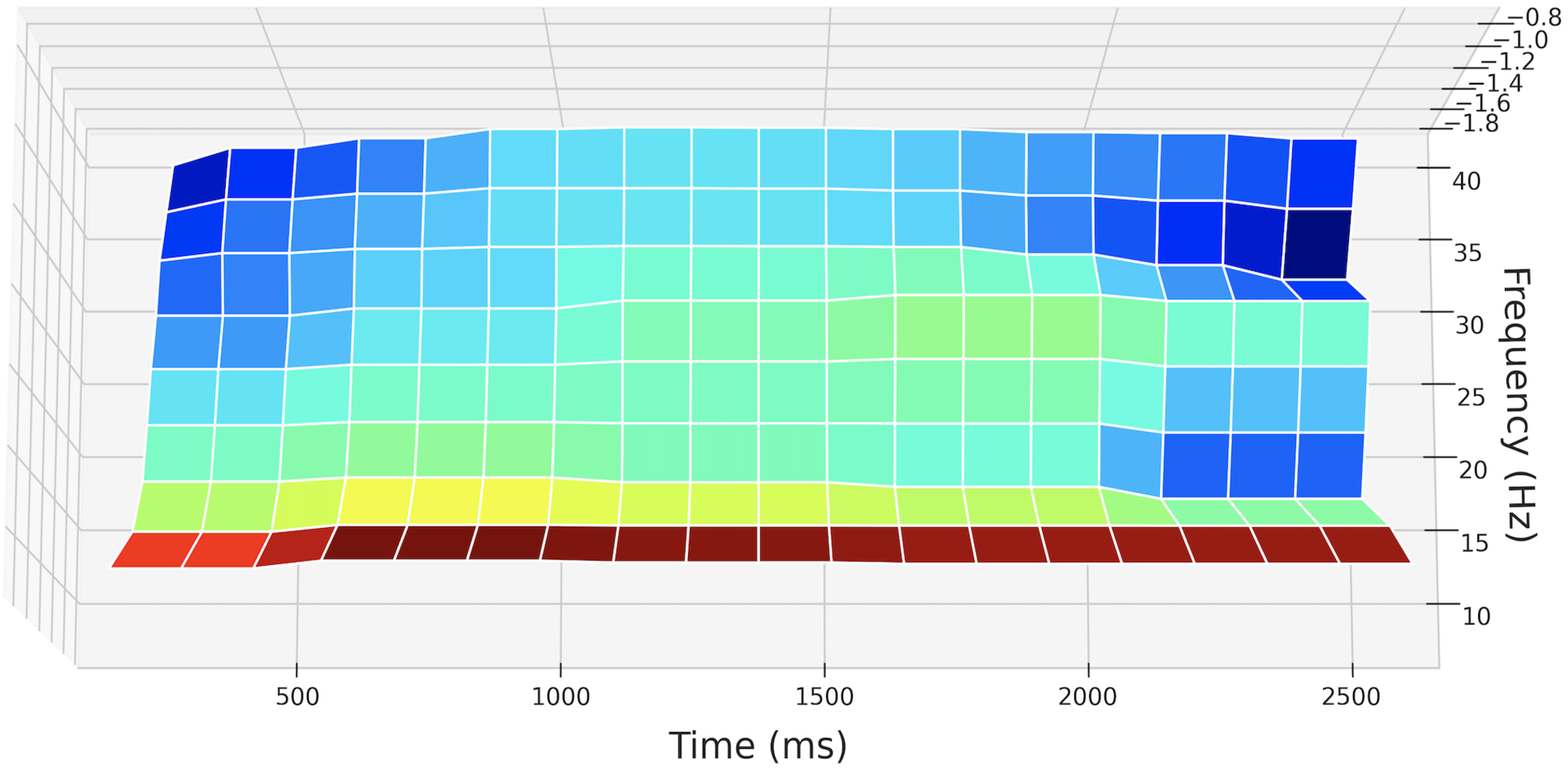}
\caption{Time Direction (2,2,2,4).}~\label{fig:c}
\end{subfigure}

\caption{Discrete Spectrograms of Variant Configuration Time-Frequency Graphs:
The LGT method has four "time direction" numbers representing the forward steps in the components of $\theta, \mu, \beta$, and $\gamma$. 
The "frequency direction" numbers are always set to Time Direction (1, 1, 4, 3).\\\\
The discrete spectrograms, ranging from (a) to (c), are calculated by evaluating the lower frequency band's spectrum power on the grids ($4\sim16$ Hz) across every five grids on the time axis, with a grid width of 500 ms and a height of 4 Hz. The higher frequency band ($16\sim40$ Hz) is calculated across each grid, with a grid width of 250 ms and a height of 4 Hz. 
Spectrogram (a) represents the original spectrum distribution of Subject No. 1 in the KU dataset, while spectrograms (b) and (c) are the spectrum distributions after the LGT method on (a).
}~\label{fig:spectrum}
\end{figure*}

\begin{figure}[!h]
\center

\begin{subfigure}{0.2\textwidth}
\includegraphics[width=\linewidth]{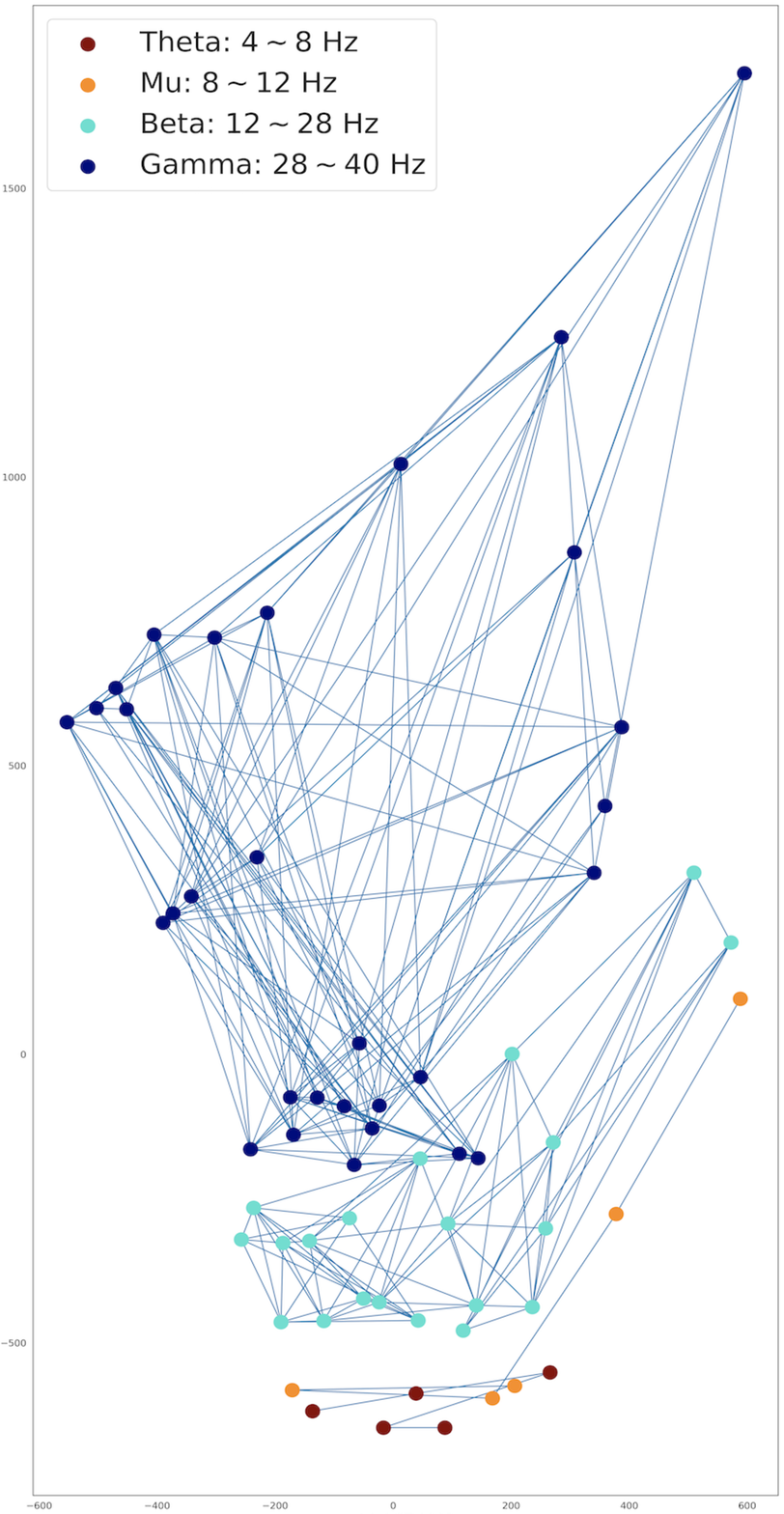}
\caption{Time Direction (1,1,1,2).}~\label{fig:a}
\end{subfigure}\hspace*{\fill}
\begin{subfigure}{0.2\textwidth}
\includegraphics[width=\linewidth]{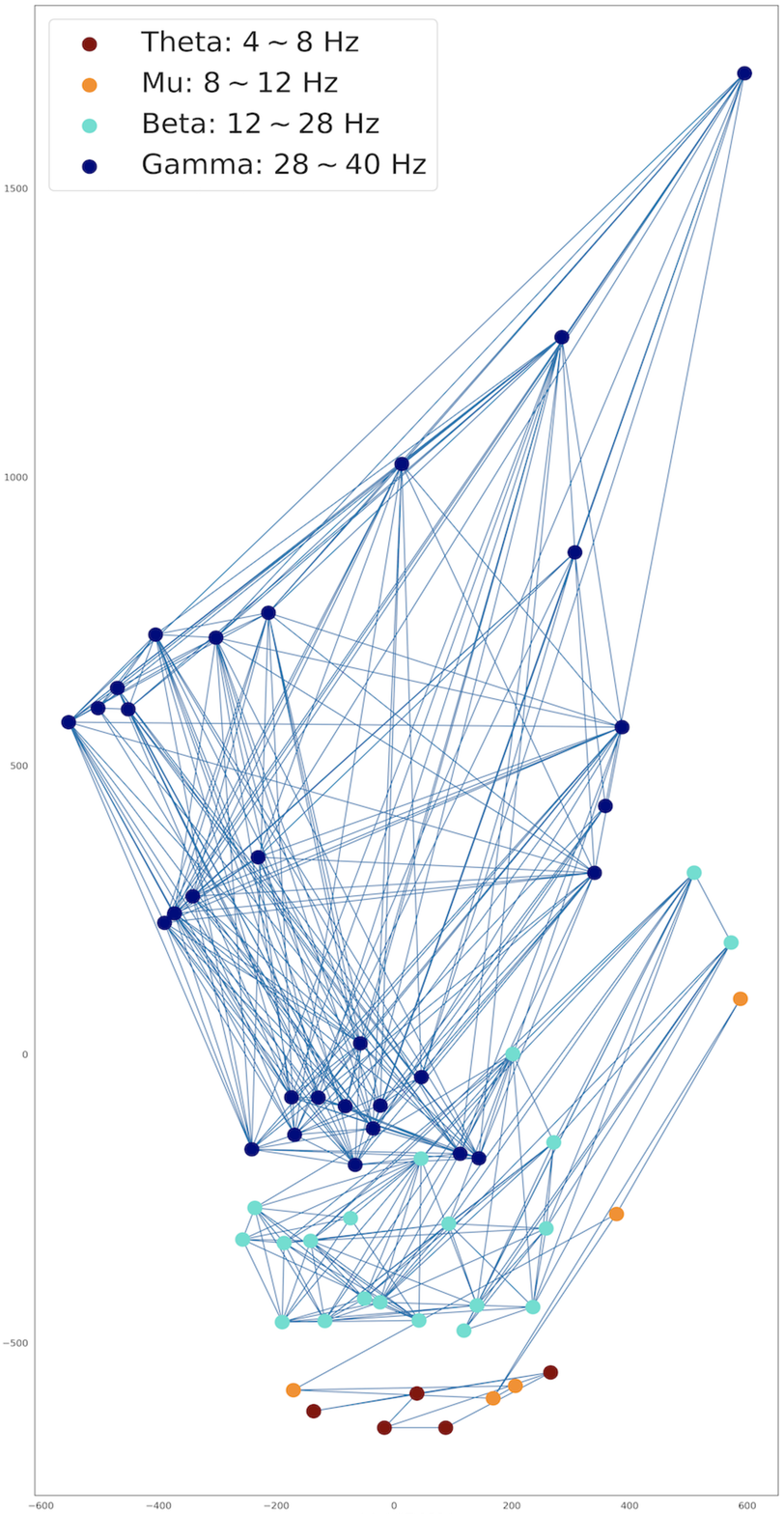}
\caption{Time Direction (2,2,2,4).}~\label{fig:b}
\end{subfigure}\hspace*{\fill}
\begin{subfigure}{0.20\textwidth}
\end{subfigure}

\caption{Variant Configurations of Time-frequency Graphs: 
The time-frequency graphs (a) and (b), which are derived from spectrograms (b) and (c), respectively, in Figure~\ref{fig:spectrum}, contain 60 nodes and 390 edges. 
Each node represents a grid or a couple of grids, with low time resolution for the low frequency (4 to 16 Hz) and high time resolution for the high frequency (16 to 40 Hz). 
The spectrum of a node evaluates adjacent nodes along the time axis and consists of four graph components corresponding to four frequency bands, i.e., $\theta, \mu, \beta$, and $\gamma$. 
The edge weight of each two adjacent nodes is the geodesic distance between the two points on ($\mathcal{S}_{++}, g^{AIRM}$) and is reconstructed using the multidimensional scaling algorithm. 
}~\label{fig:graph}
\end{figure}


\subsection{Spectrum Distribution Shift}\label{sec:spectrum_distribution_shift}

The proposed LGT method serves as a nonparametric statistical approach for initializing the topology of the time-frequency graph. According to Theorem~\ref{thm:perturbation}, both approaches will introduce changes in the distribution of the spectrum power of spatial covariance matrices.

In this section, we visually examine the extent of these changes by plotting a two-dimensional projection of the transformed spatial covariance matrices using the LGT method. This allows us to observe the shifts in their spectrum distributions.

In Figure~\ref{fig:Dist_LGT} (a), the spatial covariance matrices of the $200^{th}$ trial (represented by the bright red point) are located at the center of the cluster formed by the spatial covariance matrices of the first 199 trials (depicted in green). However, the segmented spatial covariance matrices of the $200^{th}$ trial (shown in dark red) are relatively uniformly distributed within the space occupied by the segmented spatial covariance matrices of the first 199 trials (illustrated in blue).

In contrast, Figure~\ref{fig:Dist_LGT} (b) demonstrates significant changes in the distributions of the spatial covariance matrices' spectrum. The LGT method averages each point with its neighboring points based on the graph topology. As a result, the centers of each group of points are closer together compared to the centers depicted in Figure~\ref{fig:Dist_LGT} (a).

\subsection{Variant Parameter Configurations of Time-frequency Graphs}\label{appendix:LGT}
We will investigate the spectrum distribution and resulting time-frequency graphs generated by different parameter configurations in the LGT method.

Figure~\ref{fig:spectrum} illustrates that the spectrum distribution, transformed by the LGT method, exhibits a localized forward diffusion from left to right along the forward time direction. This diffusion occurs independently in the $\theta, \mu, \beta$, and $\gamma$ components due to the separation of these frequency bands in the generation Algorithm~\ref{alg_lgt}.

Additionally, Figure~\ref{fig:graph} presents the corresponding time-frequency graphs. Nodes belonging to different frequency bands are not connected. When selecting a larger scale for the forward time flow, the number of edges in the time-frequency graph will increase according to the construction rule. For instance, the time-frequency graph in Figure~\ref{fig:graph} (b) exhibits more connected edges compared to the one in Figure~\ref{fig:graph} (a).

\section{Discussions}~\label{sec:discussion}

\subsection{Preset Edge Weights}
The edge weights of the graph are determined by calculating the average Riemannian distances between EEG segments across all samples in the training set. In a 10-fold CV scenario, adjacency matrix $A$ is initialized ten times, with slight differences between each fold. As a result, the adjacency matrix $A$ is specific to each individual, and consequently, the architecture of Graph-CSPNet that it generates is also individual-specific.

\subsection{Pseudocodes of Adjacency Matrix Generation}\label{sec:generation_A}
This section provides the pseudocode to generate an adjacency matrix $A$ of the time-frequency graph. A lattice vector is employed to acquire the average EEG spectral covariance matrices within an assigned time and frequency interval in the training set.

Taking the example of the BNCI2014001 dataset, with a 1000 ms signal and non-overlapping segmentations of 250 ms and 125 ms, we obtain 4 and 8 windows, respectively, as shown in Table~\ref{tab:non-uniform}. As a result, there are 48 EEG segments (= 4 time windows $\times$ 6 frequency bands + 8 time windows $\times$ 3 frequency bands), representing the nodes in the time-frequency graph. The lattice vector is a tensor with Dimension $= (48, 22, 22)$, where 22 is the number of electrodes. The 48 averaged spatial covariance matrices correspond to the second and third dimensions of the tensors $(i, 22, 22)$, where $1\leq i \leq 48$.

The Riemannian distances in Algorithm~\ref{alg_lgt} are calculated using the lattice vector. Specifically, the adjacency between two nodes is determined by a time-evolution mechanism. As shown in Table~\ref{tab:non-uniform}, the frequency bands are grouped into the $\theta, \mu, \beta$, and $\gamma$ components, with 4, 4, 16, and 24 nodes, respectively.

The generation algorithm takes into account the time and frequency directions in the $\theta, \mu, \beta$, and $\gamma$ components. For example, given Time Direction $(1,1,1,2)$ and Frequency Direction $(1,1,4,3)$, spatial covariance matrices evolve by one step, one step, one step, and two steps on the $\theta, \mu, \beta$, and $\gamma$ components, respectively, along the forward time direction, and one step, one step, four steps, and three steps along the frequency direction.

\begin{algorithm}[!h]
 \KwInput{Number of Nodes $(N_\theta, N_\mu, N_\beta, N_\gamma)$;\\
\hspace*{3.4em} Length of Time Horizon $(w_{\theta},w_{\mu},w_{\beta}, w_{\gamma})$; \\
\hspace*{3.4em} Time Direction $(x_{\theta},x_{\mu},x_{\beta},x_{\gamma})$; \\
\hspace*{3.4em} Frequency Direction $(y_{\theta},y_{\mu},y_{\beta},y_{\gamma})$.
 }
 \KwOutput{Adjacency Maxtrix $A[i^*, j^*]$.}

 \vspace{\baselineskip}
 Initialization of the lattice vector and start node s $\gets$ 1\;
 \vspace{\baselineskip}
 \For{$(N, w, x, y) \in [\theta, \mu, \beta, \gamma]$}{
 \vspace{\baselineskip}
 	\For{$i \gets \{1, ..., N\}$}{
		\tcc{Initialize the first coordinate for $A$.}
		$i^* \gets s + i$\;
		 \vspace{\baselineskip}
		 \tcc{Compute the largest forward time step.}
		$T \gets \min\{w - i \% w-1, x\}$\; 
		 \vspace{\baselineskip}
		\tcc{The time direction.}
		\If{$j \in \{i+1, ..., i+T+1\}$}{
			$j^* \gets s + j$\;
   			$A[i^*, j^*] \gets \exp({-d_{g^{AIRM}}^2(S_{i^*}, S_{j^*})/t})$\;
   			}
		 \vspace{\baselineskip}
		\tcc{The frequency direction.}
  		\If{$j \in \{i + wy, ..., i + wy + T+1\}$}{
			$j^* \gets s + j$\;
   			$A[i^*, j^*] \gets \exp({-d_{g^{AIRM}}^2(S_{i^*}, S_{j^*})/t})$\;
   			}
	}
\vspace{\baselineskip}	
$A[j^*, i^*] \gets A[i^*, j^*]$\;
s $\gets$ s + N.
}
\caption{Adjacency Matrix Generation~\label{alg_lgt}}
\end{algorithm}

In pseudocode, "$(N, w, x, y) \in [\theta, \mu, \beta, \gamma]$" indicates the selection of one combination from the four frequency bands ($\theta, \mu, \beta, \gamma$). For example, if $\theta$ is chosen, the combination can be represented as $(N_{\theta}, w_{\theta}, x_{\theta}, y_{\theta})$ with the $\theta$ component. To facilitate readability, we have omitted the subscripts. The "Number of Nodes $(N_\theta, N_\mu, N_\beta, N_\gamma)$" refers to the number of graph nodes belonging to a specific frequency component. The "Length of Time Horizon $(w_{\theta},w_{\mu},w_{\beta}, w_{\gamma})$" corresponds to the number of time windows in the respective frequency component, with each frequency component having different time resolutions, resulting in varying window counts. The inclusion of the maximum forward steps, denoted as $T$, is necessary because sometimes the distance to the rightmost position (the maximum time step) is less than one forward time step. The expression "$i \% w$" represents the remainder when $i$ is divided by $w$. When assigning values to the time and frequency directions, the usual convention is to first fix the frequency and move along the positive time direction. Then, within the same frequency band, higher frequencies are considered. Please refer to our GitHub repository mentioned in the abstract for specific code examples.

\subsection{Spectral Clustering and Laplacian}
Graph-CSPNet employs an SPD matrix-valued graph convolutional network to capture time-frequency information simultaneously.
Within the graph convolutional network, a technique known as spectral clustering is utilized. Spectral clustering is a multivariate statistical clustering technique that leverages the spectrum of the data's similarity matrix~\cite{von2007tutorial}. As outlined in Section~\ref{sec:graph}, the adjacency matrix $A$ of the time-frequency graph is constructed inspired by classical $\epsilon$-neighborhood neighbor methods from spectral clustering and with considerations specific to the application.
Spectral clustering enables the computation of the average spectral power across the time-frequency domain, effectively highlighting regions of interest in EEG signals for classification.

In image processing, the Laplacian operator is used to calculate the second-order derivative of an image and is effective in highlighting areas of rapid intensity changes or edges in the image. Similarly, in the context of spectral clustering in the time-frequency graph, we have a similar interpretation. By employing an overlapping segmentation approach instead of a non-overlapping one, we are able to capture more detailed local information. As the number of segments increases, spectral clustering places greater emphasis on regions of significant intensity changes on the SPD manifold, which is constructed based on statistical measures in the time-frequency domain.

Specifically, given an infinity segmentation plan that spatial covariance matrices $\{S_i\}_{i=1}^N \in \mathcal{S}_{++}^{n_C}$ of the time-frequency graph around test trial $\bar{S}$ are uniformly Gaussian distributed on SPD manifolds given in~\cite{said2017riemannian}, the discrete graph Laplacian $L_M = I- A$ of the time-frequency graph on networks-based function $f$ will converge to the continuous Laplacian $\Delta_M (f)$ with the bias term, established by many studies~\cite{hein2005graphs,singer2006graph,belkin2008towards}, as follows, $\sum_{j=1}^N L_M f(S_i)/\epsilon = \frac{1}{2} \Delta_M f(\bar{S}) + \mathcal{O}(\epsilon^{1/2})$.

\subsection{Geometric Deep Learning}

Geometric deep learning constitutes a subfield of machine learning that focuses on developing algorithms and models which are able to process non-Euclidean structured data, in particular, data represented as graphs and manifolds~\cite{bronstein2017geometric}. 

Graph Neural Networks (GNNs) are widely recognized as a primary category in geometric deep learning and have been extensively applied to solve a broad range of problems in various domains, such as traffic networks, graph-based recommender systems, molecule design, etc~\cite{wu2020comprehensive}. While the network architectures for these applications are typically designed in Euclidean spaces, there has been growing interest in GNNs in non-Euclidean spaces, or alternatively referred to as manifold-valued GNNs~\cite{chami2019hyperbolic,zhu2020graph,dai2021hyperbolic}. 

Our scenario precisely fits into the category of GNN problems in a non-Euclidean space, as each node corresponds to an EEG spatial covariance matrix, naturally represented as points on SPD manifolds. The key aspect of our problem lies in the utilization of the BiMap layer in neural networks, which incorporates the BiMap transformation to preserve the discriminative power between task classes, as characterized by the Riemannian distance (see Equation~\ref{formula_riemannian-based},~\ref{formula_AIRM}, and Table~\ref{tab:treatments}). Additionally, the graph structure provides a neural network-based framework for propagating information from different time-frequency EEG segments concurrently, guided by principles from neurophysics.

\begin{table*}[!t]
\caption{Parameters in a two-layer Graph-CSPNet with input tensor shape $(N, o_1, o_1)$.
}
\centering 
\begin{tabular}{l l l r r }
\toprule
Layer   & Type of Parameters & Shape of Outputs & Number of Parameters\\
\midrule
Graph BiMap Layer (1st) &Stiefel manifolds &$(N, o_1, o_2)$   &  $N o_1 o_2$ \\
Graph BiMap Layer (2nd)&Stiefel manifolds &$(N, o_2, o_3)$   &  $N o_2 o_3$ \\
Riemannian Batch Normalization &SPD manifolds &$(N, o_3, o_3)$   &  $o_3^2$ \\
LOG  		     &/ &$(N, o_3, o_3)$   &   / \\
Linear  	             &Euclidean &c 		  &  $cNo_3^2$ \\
\midrule
Total Number of Parameters  &   /   &  / &  $N (o_1 + o_3) o_2 + (cN+1) o_3^2$ \\
\bottomrule
\end{tabular}
\label{tab:para_architecture}
\end{table*}

\subsection{The ReEig Layer}\label{appen:projection}
Let $S$ be a real symmetric $d \times d$ matrix with eigenvalues $\lambda_1 \geq \lambda_2 \geq \cdots \geq \lambda_d$ and corresponding orthonormal eigenvectors $\{u_i\}_{i=1}^d$. Then, the spectral decomposition of $S$ is given as $S = \sum_{i=1}^d \lambda_i u_i u_i^{\top}$.
To handle the nonnegative eigenvalues in symmetric matrices, we introduce the following lemma, an essential technique in convex optimization~\cite{boyd2004convex}.
\begin{lemma}\label{l2} 
Projection $S^\dag:=\sum_{i=1}^d\,\max{\{\lambda_i, 0\}} u_i u_i^{\top}$ on the positive semidefinite cone is the extremum of the minimization problem $||S-S^\dag||_2^2$ subject to $S \succeq 0$. 
\end{lemma}

\begin{proof} 
Without loss of generality, let $\lambda_d < 0$ and define the projection $S^\dag:=\sum_{i=1}^d\max{\{\lambda_i, 0\}} u_i u_i^{\top}$. 
It follows that $||S-S^\dag||_2 = -\lambda_d$. 
To prove this lemma, it suffices to demonstrate that $||S-S^\dag||_2 \geq -\lambda_d$. 
Since we have $||S||_2 = \max{\{\lambda_1, -\lambda_d\}} = \max{\{ u_1^{\top} S u_1, -u_d^{\top} S u_d \}} = \max{ \{\sup_{||u||_2=1} u^{\top} S u, \inf_{||u||_2=1} -u^{\top} S u\}}$, then it yields that $||S-S^\dag||_2 \geq \sup_{||u||_2=1} u^{\top} (S-S^\dag) u \geq u_d^{\top} (S-S^\dag) u_d \geq -u_d^{\top} S u_d \geq -\lambda_d$.
\end{proof}

By Lemma~\ref{l2}, we obtain a means to eliminate the insignificant eigenvalues of SPD matrices. 
In practical applications, we establish a lower bound, $\beta$, with a value greater than 0 to ensure that all eigenvalues are not less than $\beta$, i.e., $S^\dag:=\sum_{i=1}^d \, \max{\{\lambda_i, \beta\}}\, u_i u_i^{\top}$. In the implementation, we consistently set the small eigenvalue as $\beta = 1e-6$. Roughly speaking, the choice of different magnitudes does not have a significant impact on the performance in the presented scenarios.

\subsection{Architecture of Graph-CSPNet}\label{appendix:architecture}
Table~\ref{tab:para_architecture} summarizes the layers and learnable network parameters in Graph-CSPNet. The total number of learnable network parameters is $N o_2 (o_1 + o_3) + (cN+1) o_3^2$. For instance, in the case of BNCI2014001, the network configuration is given by $N=48$, $c=4$, $o_1 = 22, o_2 = 36$, and $o_3 = 22$, where $o_1$ and $o_2$ are the input and output dimensions for the first graph BiMap layer respectively, $o_2$ and $o_3$ are the input and output dimensions for the second graph BiMap layer. $N$ is the number of the time-frequency graph nodes, and c is the number of classes. The total number of learnable network parameters is 169,444 parameters. Compared to Tensor-CSPNet, this amount is six times the parameters in 1-CSPNet$^{(9,1,1)}$ (27,104 parameters) and almost two-thirds of the parameters in 5-CSPNet$^{(9,3,1)}$ (232,360 parameters).

\subsection{Optimization on Smooth Manifolds}\label{sec:optimization}
Due to the parameterization of neural networks in both geometric classifiers involving manifolds, specialized optimizers that operate on manifolds are employed to enhance performance.
Specifically, the BiMap and graph BiMap layers have parameters defined on Stiefel manifolds, while the Riemannian Batch Normalization utilizes parameters defined on SPD manifolds.
To address the issue of local minima, first-order stochastic optimization methods, such as Adam and AdaGrad, which are commonly used in the Euclidean domain, need to be adapted to the corresponding Riemannian manifolds to accommodate the manifold constraints. This is crucial because the deep learning approach exhibits high non-convexity throughout its domain~\cite{becigneul2018riemannian}.

To meet this requirement, specialized first-order Riemannian adaptive optimization methods have been developed to optimize networks with manifold-valued data~\cite{geoopt2020kochurov}. 
The generalization of adaptive optimization methods involves two essential operators: retraction and parallel transport, which are established within the framework of Riemannian optimization~\cite{absil2009optimization}. 
Riemannian optimization is a field of mathematical optimization that deals with optimization problems subject to manifold constraints, i.e., $\min_{x\in \mathcal{M}} f(x)$, where $\mathcal{M}$ is a Riemannian manifold and $f:\mathcal{M}\rightarrow \mathbb{R}$ is a smooth function.
Retraction is a first-order approximation of the exponential map in the manifold space, while parallel transport refers to transporting vectors along smooth curves in manifolds equipped with an affine connection. These two operators, within the manifold domain, enable neural networks to perform gradient descent with improved precision in each update iteration.
%

\subsection{Time-Domain Statistics and Frequency-Domain Statistics}\label{TD_statistics}
In this section, we focus on exploring the similarities and differences between the statistical characteristics of EEG signals in the time and frequency domains.
Typically, the analysis of EEG signals involves extracting statistical features through spatial, temporal, and spectral analyses, either individually or in combination.
A survey paper~\cite{krusienski2012bci} classifies these features into three categories: spatial, frequency (time-frequency), and similarity.
These categories encompass various statistical measures obtained from both the time and frequency domains.

\begin{itemize}
\item Spatial Features: These features primarily involve time-domain statistics and can be obtained through supervised methods such as principal component analysis, independent component analysis, and common spatial patterns or through unsupervised methods like the common-average reference and surface Laplacian spatial filters.
\item Frequency Features: These features focus on frequency-domain statistics and include measures such as band power, fast Fourier transform, and wavelet analysis.
\item Similarity Features: This category encompasses features that capture the similarity between EEG signals, with some of these features being frequency-domain statistics. Examples of similarity features include phase-locking value and magnitude-squared 
\end{itemize}
%
This section aims to demonstrate the equivalence between time-domain statistics, represented by spatial covariance matrices, and certain frequency-domain statistics, such as band power and magnitude-squared coherence. We argue that classifiers utilizing spatial covariance matrices in MI-EEG classification can capture information comparable to that obtained from other sources, including phase locking value and magnitude-squared coherence.

This relationship has been previously explored in the literature, as discussed in~\cite{krusienski2012value}. It has been observed that various frequency-domain features often reflect the same underlying neurophysiological phenomenon during a task. However, this fundamental fact has yet to be explicitly emphasized in previous works, and its recognition is crucial for the successful implementation of geometric classifiers.

Suppose a wide-sense stationary real-valued random process $X_t$ with zero mean $\mathbb{E}(X_t) = 0$.
The auto-correlation function $R_X (\tau)$ of $X_t$ is defined as $R_X (\tau) := \mathbb{E} (X_t^\top X_{t+\tau})$, and the expected (band) power $P_X$ of $X_t$ is given by $P_X := \mathbb{E} |X_t|^2 = R_X (0) = \int_{\mathbb{R}} S_X (\omega) \, d\omega$, where $S_X (\omega)$ is the power spectral density. 
Hence, the variance of zero-centered $X_t$ is $\sigma_X^2 := \mathbb{E} |X_t - \mathbb{E}(X_t)|^2 = P_X$.
This argument applies to each channel's time series $X_c$ in multichannel EEG data. Thus, the diagonal entries of the spatial covariance matrices correspond to the variance $\sigma_{X_c}^2$, which is equivalent to the geometric classifiers expected (band) power $P_{X_c}$ of that channel $c$.
Using a similar line of reasoning based on the autocovariance function, it can be demonstrated that the off-diagonal entries of the spatial covariance matrices are equivalent to the magnitude-squared coherence without normalization.
In practice, implementing a neural network-based approach using non-zero-centered time series yields satisfactory classification performance, and there is no need to preprocess the time series by zero-centering them.

\section{Conclusions}
The primary objective of this study is to enhance the geometric classifiers in MI-EEG classification from a time-frequency analysis perspective, aiming to capture the time-frequency information of signals more effectively. Furthermore, this paper formally introduces the term \emph{geometric classifiers} as the name for this category for the first time. To achieve this goal, we present a novel network architecture called Graph-CSPNet based on an SPD matrix-valued time-frequency graph. In Graph-CSPNet, we extend graph convolutional networks to SPD manifolds equipped with $g^{AIRM}$ and utilize them to locate EEG rhythmic components in the time-frequency domain precisely. In the graph BiMap layer, we utilize preset graph weights calculated using Riemannian distances to effectively preserve discriminative information of different time-frequency spatial covariance matrices, allowing Graph-CSPNet to inherit the characteristics and advantages of several well-known MI-EEG classifiers.

In sum, the study's contribution is twofold, providing both theoretical and practical advances as follows: Theoretically, Graph-CSPNet signifies an improvement in the feasibility of existing models, approached from the fundamental principle of time-frequency analysis in signal processing.
Moreover, the potential value of this study lies in proposing a category of methods called \emph{Geometric Classifiers} that have a foundation in differential geometry within the era of deep learning. These geometric classifiers are more capable of capturing the characteristics of signal covariance matrices and offer greater possibilities associated with advanced mathematical and physical tools. By incorporating concepts from differential geometry, these methods provide a promising avenue for advancing the understanding and analysis of complex neurophysiological data, such as EEG spatial covariance matrix and fMRI functional connectivity matrix, within the framework of geometric deep learning; Practically, Graph-CSPNet provides enhanced flexibility in handling variable time-frequency resolution for signal segmentation and capturing localized fluctuations. It performs well on five widely-used MI-EEG datasets, utilizing 10-fold cross-validation and subject-specific holdout scenarios. In nine out of eleven scenarios, Graph-CSPNet achieves near-optimal results, showcasing improved classifier performance and offering valuable practical advancements.

\section*{Acknowledgments}
This study is financially supported by the RIE2020 Industry Alignment Fund–Industry Collaboration Projects (IAF-ICP) Funding Initiative, as well as cash and in-kind contributions from industry partner(s). The study is also supported by the RIE2020 AME Programmatic Fund under Singapore grant No. A20G8b0102.

{\footnotesize
\bibliographystyle{IEEEtran}
\bibliography{refs}
}

\end{document}